\documentclass{article}
\usepackage{amsmath,amssymb,amsfonts}

\def\op#1{\mathop{{\it\fam0} #1}\limits}

\newcommand{\beq}{\begin{equation}}
\newcommand{\eeq}{\end{equation}}
\newcommand{\ben}{\begin{eqnarray}}
\newcommand{\een}{\end{eqnarray}}
\newcommand{\be}{\begin{eqnarray*}}
\newcommand{\ee}{\end{eqnarray*}}
\newcommand{\bea}{\begin{eqalph}}
\newcommand{\eea}{\end{eqalph}}

\newcommand{\Ker}{{\mathrm{Ker}\,}}
\newcommand{\im}{{\mathrm{Im}\, }}

\newcommand{\al}{\alpha}
\newcommand{\bt}{\beta}
\newcommand{\dl}{\delta}
\newcommand{\la}{\lambda}

\newcommand{\f}{\phi}
\newcommand{\vf}{\varphi}

\newcommand{\om}{\omega}

\newcommand{\m}{\mu}

\newcommand{\g}{\gamma}
\newcommand{\G}{\Gamma}

\newcommand{\thh}{\theta}

\newcommand{\cG}{{\mathfrak g}}
\newcommand{\ve}{\varepsilon}

\newcommand{\e}{\epsilon}
\newcommand{\ap}{\approx}

\newcommand{\nm}[1]{|{#1}|}

\newcommand{\id}{{\mathrm{Id}\,}}
\newcommand{\re}{{\mathrm{Re}\,}}

\newcommand{\lng}{\langle}
\newcommand{\rng}{\rangle}
\newcommand{\Is}{{\mathrm{Aut}\,}}

\newcommand{\si}{\sigma}
\newcommand{\Si}{\Sigma}

\newcommand{\cA}{{\mathcal A}}

\newcommand{\cU}{{\mathcal U}}

\newcommand{\cL}{{\mathcal L}}

\newcommand{\cI}{{\mathcal I}}

\newcommand{\cE}{{\mathcal E}}
\newcommand{\cH}{{\mathcal H}}

\newcommand{\cK}{{\mathcal K}}

\newcommand{\cN}{{\mathcal N}}
\newcommand{\ccG}{{\mathfrak g}}

\newcommand{\bb}{{\mathbf 1}}

\newcommand{\wt}{\widetilde}
\newcommand{\wh}{\widehat}
\newcommand{\ol}{\overline}

\newcommand{\dr}{\partial}
\newcommand{\ar}{\op\longrightarrow}
\newcommand{\ot}{\otimes}

\let\ssection=\section
\renewcommand{\section}{\setcounter{equation}{0}\ssection}

\newenvironment{eqalph}{\stepcounter{equation}
\setcounter{equationa}{\value{equation}} \setcounter{equation}{0}

\begin{eqnarray}}{\end{eqnarray}\setcounter{equation}{\value{equationa}}}

\newcounter{equationa}[section]

\newcounter{remark}[section]
\newcounter{example}[section]
\newcounter{theorem}[section]
\newcounter{condition}[section]
\newcounter{lemma}[section]
\newcounter{proposition}[section]
\newcounter{corollary}[section]
\newcounter{definition}[section]
\def\theremark{\arabic{section}.\arabic{remark}}

\def\thetheorem{\arabic{section}.\arabic{theorem}}

\def\thedefinition{\arabic{section}.\arabic{theorem}}

\newenvironment{proof}{{\it Proof.}}{\hfill $\diamondsuit$ \medskip}
\newenvironment{remark}{\refstepcounter{remark} \medskip \noindent {\bf Remark
\theremark.} }{ \hfill $\diamondsuit$ \medskip }
\newenvironment{example}{\refstepcounter{remark} \medskip \noindent {\bf Example
\theremark.} }{ \hfill $\Box$ \medskip}
\newenvironment{theorem}{\refstepcounter{theorem} \medskip \noindent {\bf
Theorem \thetheorem.}}{\hfill $\Box$ \medskip }

\newenvironment{lemma}{\refstepcounter{theorem} \medskip \noindent {\bf  Lemma
\thetheorem.}}{\hfill $\Box$ \medskip }

\newenvironment{corollary}{\refstepcounter{theorem} \medskip \noindent {\bf
Corollary \thetheorem.} }{\hfill $\Box$ \medskip  }

\newcommand{\mar}[1]{}

\begin{document}
\hbox{}

\begin{center}

{\large\bf Inequivalent Vacuum States in Algebraic Quantum Theory}
\bigskip
\bigskip

G. Sardanashvily

\medskip

Department of Theoretical Physics, Moscow State University,
Moscow, Russia

\end{center}
\bigskip
\bigskip

\begin{abstract} The Gelfand--Naimark--Sigal representation
construction is considered in a general case of topological
involutive algebras of quantum systems, including quantum fields,
and inequivalent state spaces of these systems are characterized.
We aim to show that, from the physical viewpoint, they can be
treated as classical fields by analogy with a Higgs vacuum field.
\end{abstract}

\bigskip

\tableofcontents

\bigskip

\section{Introduction}

No long ago, one thought of vacuum in quantum field theory (QFT)
as possessing no physical characteristics, and thus being
invariant under any symmetry transformation. It is exemplified by
a particleless Fock vacuum (Section 8). Contemporary gauge models
of fundamental interactions however have arrived at a concept of
the Higgs vacuum (HV). By contrast to the Fock one, HV is equipped
with nonzero characteristics, and consequently is non-invariant
under transformations.

For instance, HV in Standard Model of particle physics is
represented by a constant background classical Higgs field, in
fact, inserted by hand into a field Lagrangian, whereas its true
physical nature still remains unclear. In particular, somebody
treats it as a {\it sui generis} condensate by analogy with the
Cooper one, and its appearance is regarded as a phase transition
characterized by some symmetry breakdown \cite{banks,bend,kawati}.

Thus, we come to a concept of different inequivalent and, in
particular, non-invariant vacua \cite{sard08a}. Here, we consider
some models of these vacua in the framework of algebraic quantum
theory (AQT). We aim to show that, from the physical viewpoint,
their characteristics are classical just as we observe in a case
of the above-mentioned Higgs vacuum.

In AQT, a quantum system is characterized by a topological
involutive algebra $A$ and a family of continuous positive forms
on $A$. Elements of $A$ are treated as quantum objects, and we
call $A$ the quantum algebra. In this framework, values of
positive forms on $A$ are regarded as numerical averages of
elements of $A$. In the spirit of Copenhagen interpretation of quantum theory, one
can think of positive forms on $A$ as being classical objects.

A corner stone of AQT is the following Gelfand--Naimark--Segal (GNS)
representation theorem \cite{book05,hor,schmu}.

\begin{theorem} \label{gns} \mar{gns}
Let $A$ be a unital topological involutive algebra and $f$ a
positive continuous form on $A$ such that $f(\bb)=1$ (i.e., $f$ is
a state). There exists a strongly cyclic Hermitian representation
$(\pi_f,\thh_f)$ of $A$ in a Hilbert space $E_f$ with a cyclic
vector $\thh_f$ such that
\mar{gns2}\beq
f(a)=\lng \pi(a)\thh_\f|\thh_\f\rng, \qquad  a\in A. \label{gns2}
\eeq
\end{theorem}

It should be emphasized that a Hilbert space $E_f$ in Theorem
\ref{gns} is a completion of the quotient of an algebra $A$ with
respect to an ideal  generated by elements $a\in A$ such that
$f(aa^*)=0$, and the cyclic vector $\thh_f$ is the image of the
identity $\bb\in A$ in this quotient. Thus, a carrier space of a
representation of $A$ and its cyclic vector in Theorem \ref{gns}
comes from a quantum algebra $A$, and they also can be treated as
quantum objects.

Since $\thh_f$ (\ref{gns2}) is a cyclic vector, we can think of it
as being a vacuum vector and, accordingly, a state $f$ as being a
vacuum of an algebra $A$. Let us note that a vacuum vector
$\thh_f$ is a quantum object, whereas a vacuum $f$ of $A$ is the
classical one. In particular, a quantum algebra $A$ acts on
quantum vectors, but not on its vacua (states).

Since vacua are classical objects, they are parameterized by
classical characteristics. A problem is that different vacua of
$A$ define inequivalent cyclic representations of a quantum
algebra $A$ in general. In this case, they are called
inequivalent.

We say that a quantum algebra $A$ performs a transition between
its vacua $f$ and $f'$ if there exist elements $b,b'\in A$ such
that $f'(a)=f(b^+ab)$ and $f(a)=f'(b'^+ab')$ for all $a\in A$. In
this case cyclic representations $\pi_f$ and $\pi_{f'}$ and,
accordingly, vacua $f$ and $f'$ are equivalent (Theorem
\ref{spr474'}). A problem thus is to characterize inequivalent
vacua of a quantum system.

One can say something in the following three variants.

(i) If a quantum algebra $A$ is a unital $C^*$-algebra, Theorem
\ref{gns} comes to well-known GNS (Theorem \ref{spr471}), and we
have a cyclic representation of $A$ by bounded operators in a
Hilbert space (Section 2). This is a case of  quantum mechanics.

(ii) A quantum algebra $A$ is a nuclear involutive algebra
(Theorem \ref{gns36}). In particular, this is just the case of
quantum field theory (Sections ?9 and 10).

(iii) Given a group $G$ of automorphisms of a quantum algebra $A$,
its vacuum is $f$ invariant only under a proper subgroup of $G$.
This is the case of spontaneous symmetry breaking in a quantum
system (Sections 12 and 13).

If a quantum algebra $A$ is a unital $C^*$-algebra, one can show
that a set $F(A)$ of states of $A$ is a weakly$^*$-closed convex
hull of a set $P(A)$ of pure states of $A$, and it is weakly*
compact (Theorem \ref{gns17}). A set $P(A)$ of pure states of $A$,
in turn, is a topological bundle over the spectrum $\wh A$ of $A$
whose fibres are projective Hilbert space. The spectrum $\wh A$ of
$A$ is a set of its nonequivalent irreducible representations
provided with the inverse image of the Jacobson topology. It is
quasi-compact.

In accordance with Theorem \ref{spr474}  a unital $C^*$-algebra
$A$ of a quantum system performs invertible transitions between
different vacua iff they are equivalent. At the same time, one can
enlarge an algebra $A$ to some algebra $B(E_F)$ so that all states
of $A$ become equivalent states of $B(E_F)$ (Theorem \ref{gns23}).
Moreover, this algebra contains the superselection operator $T$
(\ref{gns24}) which belongs to the commutant of $A$ and whose
distinct eigenvalues characterize different vacua of $A$.

In Section 4, an infinite qubit system modelled on an arbitrary
set $S$ is studied. Its quantum $C^*$-algebra $A_S$ possesses pure
states whose set is a set of maps $\si$ (\ref{gns91}) of a set $S$
to the unit sphere in $\mathbb C^2$. They are equivalent iff the
relation (\ref{gns67}) is satisfied and, in particular, if maps
$\si$ and $\si'$ differ from each other on a finite subset of $S$.
By analogy with a Higgs vacuum, one can treat the maps $\si$
(\ref{gns91}) as classical vacuum fields.

In Section 5, we consider an example of a locally compact group
$G$ and its group algebra $L^1_{\mathbb C}(G)$ of equivalence
classes of complex integrable functions on $G$. This is a Banach
involutive algebra with an approximate identity. There is
one-to-one correspondence between the representations of this
algebra and the strongly continuous unitary representations of a
group $G$ (Theorem \ref{gns61}). Continuous positive forms on
$L^1_{\mathbb C}(G)$ and, accordingly, its cyclic representations
are parameterized by continuous positive-definite functions $\psi$
on $G$ as classical vacuum fields (Theorem \ref{qm539}). If $\psi$
is square-integrable, the corresponding cyclic representation of
$L^1_{\mathbb C}(G)$ is contained in the regular representation
(\ref{qm344}). In this case, distinct square integrable continuous
positive-definite functions $\psi$ and $\psi'$ on $G$ define
inequivalent irreducible representations if they obey the
relations (\ref{qm345}).

However, this is not the case of unnormed topological
$*$-algebras. In order to say something, we restrict our
consideration to nuclear algebras (Section 6 and 7).

This technique is applied to the analysis of inequivalent
representations of infinite canonical commutative relations
(Section 8) and, in particular, free quantum fields, whose states
characterized by different masses are inequivalent.

Section 10 addresses the true functional integral formulation of
Euclidean quantum theory (Section 10). These integrals fail to be
translationally invariant that enables one to model a Higgs vacuum
a translationall inequivalent state (Section 11).

Sections 12 and 13 are devoted to the phenomenon of spontaneous
symmetry breaking when a state of a quantum algebra $A$ fails to
be stationary only with respect to some some proper subgroup $H$
of a group $G$ of automorphisms of $A$. Then a set of inequivalent
states of these algebra generated by these automorphisms is a
subset of the quotient $G/H$.

In particular, just this fact motivates us to describe classical
Higgs fields as sections of a fibre bundle with a typical fibre
$G/H$ \cite{higgs,sard08a,tmp}.

\section{GNS construction. Bounded operators}

We start with a GNS representation of a topological involutive
algebra $A$ by bounded operators in a Hilbert space. This is the
case of Banach involutive algebras with an approximate identity
(Theorem \ref{spr471}). Without a loss of generality, we however
restrict our consideration to GNS representations of
$C^*$-algebras because any involutive Banach algebra $A$ with an
approximate identity defines the enveloping $C^*$-algebra
$A^\dagger$ such that there is one-to-one correspondence between
the representations of $A$ and those of $A^\dagger$ (Remark
\ref{env}).

Let us recall the standard terminology \cite{dixm,book05}. A
complex associative algebra $A$ is called involutive (a
$^*$-algebra) if it is provided with an involution $*$ such that
\be
(a^*)^*=a, \quad (a +\la b)^*=a^* + \ol\la b^*, \quad
(ab)^*=b^*a^*, \quad  a,b\in A, \quad \la\in \mathbb C.
\ee
An element $a\in A$ is normal if $aa^*=a^*a$, and it is Hermitian
or self-adjoint if $a^*=a$. If $A$ is a unital algebra, a normal
element such that $aa^*=a^*a=\bb$ is called the unitary one.

A $^*$-algebra $A$ is called the normed algebra (resp. the Banach
algebra) if it is a normed (resp. complete normed) vector space
whose norm $\|.\|$ obeys the multiplicative conditions
\be
\| ab\| \leq\| a\|\| b\|, \qquad \| a^*\|=\| a\|, \qquad a,b\in A.
\ee
A Banach $^*$-algebra $A$ is said to be a $C^*$-algebra if
$\|a\|^2= \| a^* a\|$ for all $a\in A$. If $A$ is a unital
$C^*$-algebra, then $\|\bb\|=1$. A $C^*$-algebra is provided with
a normed topology, i.e., it is a topological $^*$-algebra.

\begin{remark}
It should be emphasized that by a morphism of normed algebras is
meant a morphism of the underlying $^*$-algebras, without any
condition on the norms and continuity. At the same time, an
isomorphism of normed algebras means always an isometric morphism.
Any morphism $\f$ of $C^*$-algebras is automatically continuous
due to the property
\mar{spr410}\beq
\|\f(a)\|\leq \| a\|, \qquad  a\in A. \label{spr410}
\eeq
\end{remark}

Any $^*$-algebra $A$ can be extended to a unital algebra $\wt
A=\mathbb C\oplus A$  by the adjunction of the identity $\bb$ to
$A$. The unital extension of $A$ also is a $^*$-algebra with
respect to the operation
\be
(\la\bb+a)^*=(\ol\la\bb+a^*), \qquad \la\in\mathbb C, \quad a\in
A.
\ee
If $A$ is a $C^*$-algebra, a norm on $A$ is uniquely prolonged to
the norm
\be
\|\la\bb +a\|= \op\sup_{\|a'\|\leq 1}\|\la a'+aa'\|
\ee
on $\wt A$ which makes $\wt A$ a $C^*$-algebra.

One says that a Banach algebra $A$ admits an approximate identity
if there is a family $\{u_\iota\}_{\iota\in I}$ of elements of
$A$, indexed by a directed set $I$, which possesses the following
properties:

$\bullet$ $\| u_\iota\|< 1$ for all $\iota\in I$,

$\bullet$ $\| u_\iota a-a\|\to 0$ and $\| au_\iota-a\|\to 0$ for
every $a\in A$.

\noindent It should be noted that the existence of an approximate
identity is an essential condition for many results (see, e.g.,
Theorems \ref{gns5} and \ref{spr471}).

For instance, a $C^*$-algebra has an approximate identity.
Conversely, any Banach $^*$-algebra $A$ with an approximate
identity admits the enveloping $C^*$-algebra $A^\dagger$ (Remark
\ref{env}) \cite{dixm,book05}.

An important example of $C^*$-algebras is an algebra $B(E)$ of
bounded (and, equivalently, continuous) operators in a Hilbert
space $E$ (Section 14.2). Every closed $^*$-subalgebra of $B(E)$
is a $C^*$-algebra and, conversely, every $C^*$-algebra is
isomorphic to a $C^*$-algebra of this type (Theorem \ref{spr424}).

An algebra $B(E)$ is endowed with the operator norm
\mar{spr427}\beq
\|a\|= \op\sup_{\|e\|_E=1} \|ae\|_E, \qquad a\in B(E).
\label{spr427}
\eeq
This norm brings the $^*$-algebra $B(E)$ of bounded operators in a
Hilbert space $E$ into a $C^*$-algebra. The corresponding topology
on $B(E)$ is called the normed operator topology.

One also provides $B(E)$ with the strong  and  weak operator
topologies,  defined by the families of seminorms
\be
&& \{p_e(a)=\|ae\|, \quad e\in E\}, \\
&& \{p_{e,e'}(a)=|\lng ae|e'\rng|, \quad e,e'\in E\},
\ee
respectively. The normed operator topology is finer than the
strong one which, in turn, is finer than the weak operator
topology. The strong and weak operator topologies on a subgroup
$U(E)\subset B(E)$ of unitary operators coincide with each other.

It should be emphasized that $B(E)$ fails to be a topological
algebra with respect to strong and weak operator topologies.
Nevertheless, the involution in $B(E)$ also is continuous with
respect to the weak operator topology, while the operations
\be
B(E)\ni a\to aa'\in B(E),\qquad B(E)\ni a\to a'a\in B(E),
\ee
where $a'$ is a fixed element of $B(E)$, are continuous with
respect to all the above mentioned operator topologies.

\begin{remark} \label{w410} \mar{w410}
Let $N$ be a subset of $B(E)$. The commutant  $N'$ of $N$ is a set
of elements of $B(E)$ which commute with all elements of $N$. It
is a subalgebra of $B(E)$. Let $N''=(N')'$ denote the bicommutant.
Clearly, $N\subset N''$. A $^*$-subalgebra $B$ of $B(E)$ is called
the von Neumann algebra if $B=B''$. This property holds iff $B$ is
strongly (or, equivalently, weakly) closed in $B(E)$ \cite{dixm}.
For instance, $B(E)$ is a von Neumann algebra. Since a strongly
(weakly) closed subalgebra of $B(E)$ also is closed with respect
to the normed operator topology on $B(E)$, any von Neumann algebra
is a $C^*$-algebra.
\end{remark}

\begin{remark} \label{w90} \mar{w90}
A bounded operator in a Hilbert space $E$ is called completely
continuous if it is compact,  i.e., it sends any bounded set into
a set whose closure is compact. An operator $a\in B(E)$ is
completely continuous iff it can be represented by the series
\mar{spr462}\beq
a(e)=\op\sum_{k=1}^\infty\la_k\lng e|e_k\rng e_k, \label{spr462}
\eeq
where $e_k$ are elements of a basis for $E$ and $\la_k$ are
positive numbers which tend to zero as $k\to\infty$. For instance,
every degenerate operator  (i.e., an operator of finite rank which
sends $E$ onto its finite-dimensional subspace) is completely
continuous. A completely continuous operator $a$ is called the
Hilbert--Schmidt operator if the series
\be
\|a\|^2_{\mathrm{HS}}= \op\sum_k \la^2_k
\ee
converges. Hilbert--Schmidt operators make up an involutive Banach
algebra with respect to this norm, and it is a two-sided ideal of
an algebra $B(E)$. A completely continuous operator $a$ in a
Hilbert space $E$ is called a nuclear operator if the series
\be
||a||_{\mathrm{Tr}}=\op\sum_k \la_k
\ee
converges. Nuclear operators make up an involutive Banach algebra
with respect to this norm, and it is a two-sided ideal of an
algebra $B(E)$. Any nuclear operator is the Hilbert--Schmidt one.
Moreover, the product of arbitrary two Hilbert--Schmidt operators
is a nuclear operator, and every nuclear operator is of this type.
\end{remark}

Let us consider representations of $^*$-algebras by bounded
operators in Hilbert spaces \cite{dixm,ped}. It is a morphism
$\pi$ of a $^*$-algebra $A$ to an algebra $B(E)$ of bounded
operators in a Hilbert space $E$, called the carrier space of
$\pi$. Representations throughout are assumed to be
non-degenerate, i.e., there is no element $e\neq 0$ of $E$ such
that $Ae=0$ or, equivalently, $AE$ is dense in $E$.

\begin{theorem} \label{spr424} \mar{spr424}
If $A$ is a $C^*$-algebra, there exists its exact (isomorphic)
representation.
\end{theorem}

\begin{theorem} \label{gns4} \mar{gns4}
A representation $\pi$ of a $^*$-algebra $A$ is uniquely prolonged
to a representation $\wt \pi$ of the unital extension $\wt A$ of
$A$.
\end{theorem}

Let $\{\pi^\iota\}$, $\iota\in I$, be a family of representations
of a $^*$-algebra $A$ in Hilbert spaces $E^\iota$. If the set of
numbers $\|\pi^\iota(a)\|$ is bounded for each $a\in A$, one can
construct a bounded operator $\pi(a)$ in a Hilbert sum $\oplus
E^\iota$ whose restriction to each $E^\iota$ is $\pi^\iota(a)$.

\begin{theorem} \label{sum} \mar{sum}
This is the case of a $C^*$-algebra $A$ due to the property
(\ref{spr410}). Then $\pi$ is a representation of $A$ in $\oplus
E^\iota$, called the Hilbert sum
\mar{0315}\beq
\pi=\op\oplus_I \pi^\iota \label{0315}
\eeq
of representations $\pi^\iota$.
\end{theorem}

Given a representation $\pi$ of a $^*$-algebra $A$ in a Hilbert
space $E$, an element $\thh\in E$ is said to be a cyclic vector
for $\pi$  if the closure of $\pi(A)\thh$ is equal to $E$.
Accordingly, $\pi$ (or a more strictly a pair $(\pi,\thh)$) is
called the cyclic representation.

\begin{theorem} \label{spr411} \mar{spr411}
Every representation of a $^*$-algebra $A$ is a Hilbert sum of
cyclic representations.
\end{theorem}

\begin{remark} \label{gns29} \mar{gns29}
It should be emphasized, that given a cyclic representation
$(\pi,\thh)$ of a $^*$-algebra $A$ in a Hilbert space $E$, a
different element $\thh'$ of $E$ is a cyclic for $\pi$ iff there
exist some elements $b, b'\in A$ such that $\thh'=\pi(b)\thh$ and
$\thh=\pi(b')\thh'$.
\end{remark}

Let $A$ be a $^*$-algebra, $\pi$ its representation in a Hilbert
space $E$, and $\thh$ an element of $E$. Then a map
\mar{spr470}\beq
\om_\thh: a\to \lng \pi(a)\thh|\thh\rng \label{spr470}
\eeq
is a positive form on $A$. It is called the vector form defined by
$\pi$ and $\thh$.

Therefore, let us consider positive forms on a $^*$-algebra $A$.
Given a positive form $f$, a Hermitian form
\mar{spr472}\beq
\lng a|b\rng = f(b^*a), \qquad a,b\in A, \label{spr472}
\eeq
makes $A$ a pre-Hilbert space. If $A$ is a normed $^*$-algebra,
continuous positive forms on $A$ are provided with a norm
\mar{0215}\beq
 \|f\|=\op\sup_{\|a\|=1}|f(a)|, \qquad a\in A. \label{0215}
\eeq

\begin{theorem} \label{gns6} \mar{gns6} Let $A$ be a unital Banach
$^*$-algebra such that $\|\bb\|=1$. Then any positive form on $A$
is continuous.
\end{theorem}

In particular, positive forms on a $C^*$-algebra always are
continuous. Conversely, a continuous form $f$ on an unital
$C^*$-algebra is positive iff $f(\bb)=\|f\|$. It follows from this
equality that positive forms on a unital $C^*$-algebra $A$ obey a
relation
\mar{gns12}\beq
\|f_1+f_2\|=\|f_1\|+ \|f_2\|. \label{gns12}
\eeq

Let us note that a continuous positive form on a topological
$^*$-algebra $A$ admits different prolongations onto the unital
extension $\wt A$ of $A$. Such a prolongation is unique in the
following case \cite{dixm}.

\begin{theorem} \label{gns5} \mar{gns5}
Let $f$ be a positive form on a Banach $^*$-algebra $A$ with an
approximate identity. It is extended to a unique positive form
$\wt f$ on the unital extension $\wt A$ of $A$ such that $\wt
f(\bb)=\|f\|$.
\end{theorem}

A key point is that any positive form on a $C^*$-algebra equals a
vector form defined by some cyclic representation of $A$ in
accordance with the following GNS representation construction
\cite{dixm,book05}.

\begin{theorem} \label{spr471} \mar{spr471}
Let $f$ be a positive form on a Banach $^*$-algebra $A$ with an
approximate identity and $\wt f$ its continuous positive
prolongation onto the unital extension $\wt A$ (Theorems
\ref{gns6} and \ref{gns5}). Let $N_f$ be a left ideal of $\wt A$
consisting of those elements $a\in A$ such that $\wt f(a^*a)=0$.
The quotient $\wt A/N_f$ is a Hausdorff pre-Hilbert space with
respect to the Hermitian form obtained from $\wt f(b^*a)$
(\ref{spr472}) by passage to the quotient. We abbreviate with
$E_f$ the completion of $\wt A/N_f$ and with $\thh_f$ the
canonical image of $\bb\in \wt A$ in $\wt A/N_f\subset E_f$. For
each $a\in \wt A$, let $\tau(a)$ be an operator in $\wt A/N_f$
obtained from the left multiplication by $a$ in $\wt A$ by passage
to the quotient. Then the following holds.

(i) Each $\tau(a)$ has a unique extension to a bounded operator
$\pi_f(a)$ in a Hilbert space $E_f$.

(ii) A map $a\to \pi_f(a)$ is a representation of $A$ in $E_f$.

(iii) A representation $\pi_f$ admits a cyclic vector $\thh_f$.

(iv) $f(a)=\lng \pi(a)\thh_f|\thh_f\rng$ for each $a\in A$.
\end{theorem}

The representation $\pi_f$ and the cyclic vector $\thh_f$ in
Theorem \ref{spr471} are said to be defined by a form $f$, and
a form $f$ equals the vector form defined by $\pi_f$ and
$\thh_f$.

As was mentioned above, we further restrict our consideration of
the GNS construction in Theorem \ref{spr471} to unital
$C^*$-algebras in view of the following \cite{dixm,book05}.

\begin{remark} \label{env} \mar{env}
Let $A$ be an involutive Banach algebra $A$ with an approximate
identity, and let $P(A)$ be the set of pure states of $A$ (Remark
\ref{gns60}). For each $a\in A$, we put
\mar{spr542}\beq
\|a\|'=\op\sup_{f\in P(A)} f(aa^*)^{1/2}, \qquad a\in A.
\label{spr542}
\eeq
It is a seminorm on $A$ such that $\|a\|'\leq \|a\|$. If $A$ is a
$C^*$-algebra, $\|a\|'= \|a\|$ due to the relation (\ref{spr410})
and the existence of an isomorphic representation of $A$. Let
$\cI$ denote the kernel of $\|.\|'$. It consists of $a\in A$ such
that $\|a\|'=0$. Then the completion $A^\dagger$ of the factor
algebra $A/\cI$ with respect to the quotient of the seminorm
(\ref{spr542}), is a $C^*$-algebra, called the enveloping
$C^*$-algebra of $A$.   There is the canonical morphism $\tau:A\to
A^\dagger$. Clearly, $A=A^\dagger$ if $A$ is a $C^*$-algebra. The
enveloping $C^*$-algebra $A^\dagger$ possesses the following
important properties.

$\bullet$ If $\pi$ is a representation of $A$, there is exactly
one representation $\pi^\dagger$ of $A^\dagger$ such that
$\pi=\pi^\dagger\circ\tau$. Moreover, the map $\pi\to\pi^\dagger$
is a bijection of a set of representations of $A$ onto a set of
representations of $A^\dagger$.

$\bullet$  If $f$ is a continuous positive form on $A$, there
exists exactly one positive form $f^\dagger$ on $A^\dagger$ such
that $f=f^\dagger\circ\tau$. Moreover, $\|f^\dagger\|=\|f\|$. The
map $f\to f^\dagger$ is a bijection of a set of continuous
positive forms on $A$ onto a set of positive forms on $A^\dagger$.
\end{remark}

Moreover, the cyclic vector $\thh_f$ in Theorem \ref{spr471}
defined by a positive form $f$ is the image of the identity under
the quotient $\wt A\to\wt A/N_f$, and thus the GNS construction
necessarily is concerned with unital algebras. In view of Theorems
\ref{gns4} and \ref{gns5}, we therefore can restrict our
consideration to unital $C^*$-algebras.

\section{Inequivalent vacua}

Let $A$ be a unital $C^*$-algebra of a quantum systems. As was
mentioned above, positive forms on a $C^*$ algebra are said to be
equivalent if they define its equivalent cyclic representations.

\begin{remark} \label{gns122} \mar{gns122}
Let us recall that two representations $\pi_1$ and $\pi_2$ of a
$^*$-algebra $A$ in Hilbert spaces $E_1$ and $E_2$ are equivalent
if there is an isomorphism $\g:E_1\to E_2$ such that
\mar{gns10}\beq
\pi_2(a)=\g\circ\pi_1(a)\circ \g^{-1}, \qquad a\in A.
\label{gns10}
\eeq
In particular, if representations are equivalent, their kernels
coincide with each other.
\end{remark}

Given two positive forms $f_1$ and $f_2$ on a unital $C^*$-algebra
$A$, we meet the following three variants.

(i) If $f_1=f_2$, there is an isomorphism $\g$ of the
corresponding Hilbert spaces $\g:E_1\to E_2$ such that the
relation (\ref{gns10}) holds,  and moreover
\mar{gns11}\beq
\thh_2=\g(\thh_1). \label{gns11}
\eeq

(ii) Let positive forms $f_1$ and $f_2$ be equivalent, but
different. Then their equivalence morphism $\g$ fails to obey the
relation (\ref{gns11}).

(iii) Positive forms $f_1$ and $f_2$ on $A$ are inequivalent.

In particular, let $\pi$ be a representation of $A$ in a Hilbert
space $E$, and let $\thh$ be an element of $E$ which defines the
vector form $\om_\thh$ (\ref{spr470}) on $A$. Then a
representation $\pi$ contains a summand which is equivalent to the
cyclic representation $\pi_{\om_\thh}$ of $A$ defined by a vector
form $\om_\thh$.

There are the following criteria of equivalence of positive forms.

\begin{theorem} \label{gns123} \mar{gns123}
Positive forms on a unital $C^*$-algebra are equivalent only if
their kernels contain a common largest closed two-sided ideal.
\end{theorem}

\begin{proof}
The result follows from the fact that the kernel of a cyclic
representation defined by a positive form on a unital
$C^*$-algebra is a largest closed two-sided ideal of the kernel of
this form \cite{dixm}
\end{proof}

\begin{theorem} \label{spr474} \mar{spr474}
Positive forms $f$ and $f'$ on a unital $C^*$-algebra $A$ are
equivalent iff there exist elements $b,b'\in A$ such that
\be
f'(a)=f(b^+ab), \qquad f(a)=f'(b'^+ab'), \qquad a\in A.
\ee
\end{theorem}

\begin{proof} Let a positive form $f$ define a cyclic representation $(\pi_f,\thh_f)$ of $A$
in $E_f$. Let us consider an element $\pi_f(b)\thh_f\in E_f$. In
accordance with Remark \ref{gns29}, this element is a cyclic
element for a representation $\pi_f$. It provides a positive form
$\om_{\pi_f(b)\thh_f}$ on $A$ such that $\om_{\pi_f(b)\thh_f}=f'$.
Then a positive form $f'$ defines a cyclic representation
$(\pi_{f'},\thh_{f'})$ of $A$ in a Hilbert space $E_{f'}$ which is
isomorphic as $\g:E_{f'}\to E_f$ to a cyclic representation
$(\pi_f,\pi_f(b)\thh_f)$ in $E_f$ such that the relation
(\ref{gns11}) holds. Conversely, let positive forms $f$ and $f'$
be equivalent. Then a positive form $f'$ defines an isomorphic
cyclic representation $(\pi_f, \thh')$ in $E_f$, but with a
different cyclic vector $\thh'$. Then the result follows from
Remark \ref{gns29}.
\end{proof}

In particular, it follows from Theorem \ref{spr474} that given a
positive form $f$ on $A$, the state $f(\bb)^{-1}f$ of $A$ is
equivalent to $f$. Speaking on equivalent positive forms on $A$,
we therefore can restrict our consideration to states.

For instance, any cyclic representation of a $C^*$-algebra $A$ is
a summand of the Hilbert sum (\ref{0315}):
\mar{gns13}\beq
\pi_F=\op\oplus_{F(A)} \pi_f, \label{gns13}
\eeq
of cyclic representations of $A$ where $f$ runs through a set
$F(A)$ of all states of $A$. Since for any element $a\in A$ there
exists a state $f$ such that $f(a)\neq 0$, the representation
$\pi_F$ is injective and, consequently, isometric and isomorphic.

A space of continuous forms on a $C^*$-algebra $A$ is the
(topological) dual $A'$ of a Banach space $A$. It can be provided
both with a normed topology defined by the norm (\ref{0215}) and a
weak$^*$ topology (Section 14.1). It follows from the relation
(\ref{gns12}), that a subset $F(A)\subset A'$ of states is convex
and its extreme points are pure states.

\begin{remark} \mar{gns60} \label{gns60}
Let us recall that a positive form $f'$ on a $^*$-algebra $A$ is
said to be dominated  by a positive form $f$ if $f-f'$ is a
positive form \cite{book05,dixm}. A non-zero positive form $f$ on
a $^*$-algebra $A$ is called pure if every positive form $f'$ on
$A$ which is dominated by $f$ reads $\la f$, $0\leq \la\leq 1$.
\end{remark}

A key point is the following \cite{dixm}

\begin{theorem} \label{spr480} \mar{spr480}
The cyclic representation of $\pi_f$ of a $C^*$-algebra $A$
defined by a positive form $f$ on $A$ is irreducible iff $f$ is a
pure form \cite{dixm}
\end{theorem}

In particular, any vector form defined by a vector of a carrier
Hilbert space of an irreducible representation is a pure form.

\begin{remark}
Let us note that a representation $\pi$ of a $^*$-algebra $A$ in a
Hilbert space $E$ is called topologically irreducible if the
following equivalent conditions hold:

$\bullet$ the only closed subspaces of $E$ invariant under
$\pi(A)$ are 0 and $E$;

$\bullet$ the commutant of $\pi(A)$ in $B(E)$ is a set of scalar
operators;

$\bullet$ every non-zero element of $E$  is a cyclic vector for
$\pi$.

\noindent At the same time, irreducibility of $\pi$ in the
algebraic sense means that the only subspaces of $E$ invariant
under $\pi(A)$ are 0 and $E$. If $A$ is a $C^*$-algebra, the
notions of topologically and algebraically  irreducible
representations are equivalent. It should be emphasized that a
representation of a $C^*$-algebra need not be a Hilbert sum of the
irreducible ones.
\end{remark}

An algebraically irreducible representation $\pi$ of a
$^*$-algebra $A$ is characterized by its kernel $\Ker\pi\subset
A$. This is a two-sided ideal, called primitive. Certainly,
algebraically irreducible representations with different kernels
are inequivalent, whereas equivalent irreducible representations
possesses the same kernel. Thus, we have a surjection
\mar{spr736}\beq
\wh A\ni \pi\to \Ker\pi \in \mathrm{Prim}(A) \label{spr736}
\eeq
of a set $\wh A$ of equivalence classes of algebraically
irreducible representations of a $^*$-algebra $A$ onto a set
Prim$(A)$ of primitive ideals of $A$.

A set Prim$(A)$ is equipped with the so called Jacobson topology
\cite{dixm}. This topology is not Hausdorff, but obeys the
Fr\'echet axiom,  i.e., for any two distinct points of Prim$(A)$,
there is a neighborhood of one of them which does not contain the
other. Then a set $\wh A$ is endowed with the coarsest topology
such that the surjection (\ref{spr736}) is continuous. Provided
with this topology, $\wh A$ is called the spectrum  of a
$^*$-algebra $A$. In particular, one can show the following.

\begin{theorem} \label{spr413} \mar{spr413}
If a $^*$-algebra $A$ is unital, its spectrum $\wh A$ is
quasi-compact, i.e., it satisfies the Borel--Lebesgue axiom, but
need not be Hausdorff.
\end{theorem}

\begin{theorem} \label{spr414} \mar{spr414}
The spectrum $\wh A$ of a $C^*$-algebra $A$ is a locally
quasi-compact space.
\end{theorem}

It follows from Theorems \ref{spr413} and \ref{spr414} that the
spectrum of a unital $C^*$-algebra is quasi-compact.

\begin{example} \label{spr770} \mar{spr770} A $C^*$-algebra is said to be
elementary if it is isomorphic to an algebra $T(E)\subset B(E)$ of
compact operators in some Hilbert space $E$ (Example \ref{w90}).
Every non-trivial irreducible representation of an elementary
$C^*$ algebra $A\cong T(E)$ is equivalent to its isomorphic
representation by compact operators in $E$ \cite{dixm}. Hence, the
spectrum of an elementary algebra is a singleton set.
\end{example}

By analogy with Theorem \ref{spr474}, one can state the following
relations between equivalent pure states of a $C^*$-algebra.

\begin{theorem} \label{spr481} \mar{spr481}
Pure states $f$ and $f'$ of a unital $C^*$-algebra $A$ are
equivalent iff there exists a unitary element $U\in A$ such that
the relation
\mar{gns51'}\beq
f'(a)=f(U^*aU), \qquad  a\in A. \label{gns51'}
\eeq
holds.
\end{theorem}

\begin{proof}
A key point is that, if $f$ is a pure state of a unital
$C^*$-algebra, a pseudo-Hilbert space $A/N_f$ in Theorem
\ref{spr471} is complete, i.e., $E_f=A/N_f$.
\end{proof}

\begin{corollary} \label{gns33} \mar{gns33}
Let $\pi$ be an irreducible representation of a unital
$C^*$-algebra $A$ in a Hilbert space $E$. Given two distinct
elements $\thh_1$ and $\thh_2$ of $E$ (they are cyclic for $\pi$),
the vector forms on $A$ defined by $(\pi,\thh_1)$ and $(\pi,
\thh_2)$ are equal iff there exists $\la\in\mathbb C$, $|\la|=1$,
such that $\thh_1=\la\thh_2$.
\end{corollary}

\begin{corollary} \label{gns34} \mar{gns34}
There is one-to-one correspondence between the pure states of a
unital $C^*$-algebra $A$ associated to the same irreducible
representation $\pi$ of $A$ in a Hilbert space $E$ and the
one-dimensional complex subspaces of $E$, i.e, these pure states
constitute a projective Hilbert space $PE$.
\end{corollary}

There is an additional important criterion of equivalence of pure
states of a unital $C^*$-algebra \cite{glimm}.

\begin{theorem} \label{gns18} \mar{gns18}
Pure states $f$ and $f'$ of a unital $C^*$-algebra are equivalent
if $\|f-f'\|<2$.
\end{theorem}

Let $P(A)$ denote a set of pure states of a unital $C^*$-algebra
$A$. Theorem \ref{spr481} implies a surjection $P(A)\to \wh A$.
One can show that, if $P(A)\subset A'$ is provided with a relative
weak$^*$ topology, this surjection is continuous and open, i.e.,
it is a topological fibre
 bundle whose fibres are projective
Hilbert spaces \cite{dixm}.

Turning to a set $F(A)$ of states of a unital algebra
$C^*$-algebra $A$, we have the following.

\begin{theorem} \label{gns17} \mar{gns17}
A set $F(A)$ is a weakly$^*$-closed  convex hull of a set $P(A)$
of pure states of $A$. It is weakly* compact \cite{dixm}.
\end{theorem}

Herewith, by virtue of Theorem \ref{gns17}, any set of mutually
inequivalent pure states of a unital $C^*$-algebra is totally
disconnected in a normed topology, i.e., its connected components
are points only.

By virtue of Theorem \ref{spr474}, elements of a quantum algebra
$A$ can not perform invertible transitions between its
inequivalent states. At the same time, one can show the following.

\begin{theorem} \label{gns23} \mar{gns23}
There exists a wider unital $C^*$-algebra such that inequivalent
states of $A$ become its equivalent ones.
\end{theorem}

\begin{proof}
Let us consider the Hilbert sum $\pi_F$ (\ref{gns13}) of cyclic
representations of $A$ whose carrier space is a Hilbert sum
\mar{gns21}\beq
E_F=\op\oplus_{F(A)}E_f. \label{gns21}
\eeq
Let $B(E_F)$ be a unital $C^*$ algebra of bounded operators in
$E_F$ (\ref{gns21}). Since the representation $\pi_F$ of $A$ is
exact, an algebra $A$ is isomorphic to a subalgebra of $B(E_F)$.
Any state $f$ of $A$ is equivalent to a vector state
$\om_{\thh_f}$ of $\pi_f(A)$ which also is that of $B(E_F)$. Since
all vector states of $B(E_F)$ are equivalent, all states of $A$
are equivalent as those of $B(E_F)$.
\end{proof}

An algebra $B(E_F)$ contains the projectors $P_f$ onto summands
$E_f$ of $E_F$ (\ref{gns21}). Let $r(f)$ be some real function on
$F(A)$. Then there exists a bounded operator in $E_F$
(\ref{gns21}), which we denote
\mar{gns24}\beq
T=\op\sum_{F(A)}r(f)P_f, \label{gns24}
\eeq
such that its restriction to each summand $E_f$ of $E_F$ is
$r(f)P_f$. Certainly, this operator belongs to the commutant
$\pi_F(A)'$ of $\pi_F(A)$ in $B(E_F)$.

One can think of $T$ (\ref{gns24}) as being a superselection
operator of a quantum system which distinguish its states
\cite{hor}.

\section{Example. Infinite qubit systems}

Let $Q$ be a two-dimensional complex space $\mathbb C^2$ equipped
with the standard positive non-degenerate Hermitian form
$\lng.|.\rng_2$. Let $M_2$ be an algebra of complex $2\times
2$-matrices seen as a $C^*$-algebra. A system of $m$ qubits
 is usually described by a Hilbert space
$E_m=\op\ot^m Q$ and a $C^*$-algebra $A_m=\op\ot^m M_2$, which
coincides with the algebra $B(E_m)$ of bounded operators in $E_m$
\cite{keyl}. One can straightforwardly generalize this description
to an infinite set $S$ of qubits by analogy with a spin lattice
\cite{emch,book05,qubit}. Its algebra $A_S$ admits inequivalent
irreducible representations.

We follow the construction of infinite tensor products of Hilbert
spaces and $C^*$-algebras in \cite{emch}. Let $\{Q_s,s\in S\}$ be
a set of two-dimensional Hilbert spaces $Q_s=\mathbb C^2$. Let
$\op\times_SQ_s$ be a complex vector space whose elements are
finite linear combinations of elements $\{q_s\}$ of the Cartesian
product $\op\prod_S Q_s$ of the sets $Q_s$. A tensor product
$\op\ot_S Q_s$ of complex vector spaces $Q_s$ is the quotient of
$\op\times_S Q_s$ with respect to a vector subspace generated by
elements of the form:

$\bullet$ $\{q_s\} + \{q'_s\} -\{q''_s\}$, where $q_r  + q'_r
=q''_r$ for some element $r\in S$ and $q_s = q'_s =q''_s$ for all
the others,

$\bullet$ $\{q_s\} - \la\{q'_s\}$, $\la\in\mathbb C$, where $q_r=
\la q'_r$ for some element $r\in S$ and $q_s = q'_s$ for all the
others.

Given a map
\mar{gns91}\beq
\si:S\to Q, \qquad , \lng\si(s)|\si(s)\rng_2=1, \label{gns91}
\eeq
let us consider an element
\mar{gns92}\beq
\thh_\si=\{\thh_s=\si(s)\}\in \op\prod_S Q_s. \label{gns92}
\eeq
Let us denote $\ot^\si Q_s$ the subspace of $\op\ot_S Q_s$ spanned
by vectors $\ot q_s$ where $q_s\neq \thh_s$ only for a finite
number of elements $s\in S$. It is called the $\thh_\si$-tensor
product of vector spaces $Q_s$, $s\in S$.  Then $\ot^\si Q_s$ is a
pre-Hilbert space with respect to a positive non-degenerate
Hermitian form
\be
\lng \ot^\si q_s|\ot^\si q'_s\rng=\op\prod_{s\in S} \lng
q_s|q'_s\rng_2.
\ee
Its completion $Q_S^\si$ is a Hilbert space whose orthonormal
basis consists of the elements $e_{ir}=\ot q_s$, $r\in S$,
$i=1,2$, such that $q_{s\neq r}=\thh_s$ and $q_r=e_i$, where
$\{e_i\}$ is an orthonormal basis for $Q$.

Let now $\{A_s,s\in S\}$ be a set of unital $C^*$-algebras
$A_s=M_2$. These algebras are provided with the operator norm
\be
\|a\|=(\la_0\ol\la_0 +\la_1\ol\la_1 +\la_2\ol\la_2
+\la_3\ol\la_3)^{1/2}, \qquad a=i\la_0\bb
+\op\sum_{i=1,2,3}\la_i\tau^i,
\ee
where $\tau^i$ are the Pauli matrices. Given the family
$\{\bb_s\}$, let us construct the $\{\bb_s\}$-tensor product $\ot
A_s$ of vector spaces $A_s$. One can regard its elements as tensor
products of elements of $a_s\in A_s$, $s\in K$, for finite subsets
$K$ of $S$ and of the identities $\bb_s$, $s\in S\setminus K$. It
is easily justified that $\ot A_s$ is a normed $^*$-algebra with
respect to the operations
\be
(\ot a_s)(\ot a'_s)=\ot (a_sa'_s), \qquad (\ot a_s)^*=\ot a^*_s
\ee
and a norm
\be
\|\ot a_s\|=\op\prod_s\|a_s\|.
\ee
Its completion $A_S$ is a $C^*$-algebra treated as a quantum
algebra of a qubit system modelled over a set $S$. Then the
following holds \cite{emch}.

\begin{theorem} \label{gns63} \mar{gns63}
Given the element $\thh_\si=\{\thh_s\}$ (\ref{gns92}), the natural
representation $\pi^\si$ of a $^*$-algebra $\ot A_s$ in the
pre-Hilbert space $\ot^\si Q_s$ is extended to an irreducible
representation of a $C^*$-algebra $A_S$ in the Hilbert space
$Q_S^\si$ such that $\pi^\si(A_S)=B(Q_S^\si)$ is an algebra of all
bounded operators in $Q_S^\si$. Conversely, all irreducible
representations of $A_S$ are of this type.
\end{theorem}

An element $\thh_\si\in Q_S^\si$ in Theorem \ref{gns63} defines a
pure state $f_\si$ of an algebra $A_S$. Consequently, a set of
pure states of this algebra is a set of maps $\si$ (\ref{gns91}).

\begin{theorem}
Pure states $f_\si$ and $f_{\si'}$ of an algebra $A_S$ are
equivalent iff
\mar{gns67}\beq
\op\sum_{s\in S} ||\lng\si(s)|\si'(s)\rng_2|-1|<\infty.
\label{gns67}
\eeq
\end{theorem}

In particular, the relation (\ref{gns67}) holds if maps $\si$ and
$\si'$ differ from each other on a finite subset of $S$.

By analogy with a Higgs vacuum, one can treat the maps $\si$
(\ref{gns91}) as classical vacuum fields.

\section{Example. Locally compact groups}

Let $G$ be a locally compact group provided with a Haar measure
(Section 14.4). A space $L^1_\mathbb C(G)$ of equivalence classes
of complex integrable functions (or, simply, complex integrable
functions) on $G$ is an involutive Banach algebra (Section 14.3)
with an approximate identity. As was mentioned above, there is
one-to-one correspondence between the representations of this
algebra and the strongly continuous unitary representations of a
group $G$ (Theorem \ref{gns61}). Thus, one can employ the GNS
construction in order to describe these representations of $G$
\cite{dixm,book05}.

Let a left Haar measure $dg$ on $G$ hold fixed, and by an
integrability condition throughout is meant the
$dg$-integrability.

A  uniformly (resp. {strongly) continuous unitary representation}
of a locally compact group $G$ in a Hilbert space $E$ is a
continuous homeomorphism $\pi$ of $G$ to a subgroup $U(E)\subset
B(E)$ of unitary operators in $E$ provided with the normed (resp.
strong) operator topology. A uniformly continuous representation
is strongly continuous. However, the uniform continuity of a
representation is rather rigorous condition. For instance, a
uniformly continuous irreducible unitary representation of a
connected locally compact real Lie group is necessarily
finite-dimensional. Therefore, one usually studies strongly
continuous representations of locally compact groups.

In this case, any element $\xi$ of a carrier Hilbert space $E$
yields the continuous map $G\ni g\to \pi(g)\xi\in E$. Since strong
and weak operator topologies on a unitary group $U(E)$ coincide,
we have a bounded continuous complex function
\mar{spr530}\beq
\vf_{\xi,\eta}(g)=\lng \pi(g)\xi|\eta\rng \label{spr530}
\eeq
on $G$ for any fixed elements $\xi,\eta\in E$. It is called the
coefficient of a representation $\pi$. There is an obvious
equality
\be
\vf_{\xi,\eta}(g)=\ol{\vf_{\eta,\xi}(g^{-1})}.
\ee

The Banach space $L^1_\mathbb C(G)$ of integrable complex
functions on $G$ is provided with the structure of an involutive
Banach algebra with respect to the contraction $f_1*f_2$
(\ref{spr531}) and the involution
\be
f(g)\to f^*(g)=\Delta(g^{-1})\ol{f(g^{-1})},
\ee
where $\Delta$ is the modular function of $G$. It is called the
group algebra of $G$. A map $f\to f(g)dg$ defines an isometric
monomorphism of $L^1_\mathbb C(G)$ to a Banach algebra $M^1(G,
\mathbb C)$ of bounded complex measures on $G$ provided with the
involution $\m^*=\ol{\m^{-1}}$. Unless otherwise stated,
$L^1_\mathbb C(G)$ will be identified with its image in
$M^1(G,\mathbb C)$. In particular, a group algebra $L^1_\mathbb
C(G)$ admits an approximate identity which converges to the Dirac
measure $\ve_\bb\in M^1(G,\mathbb C)$.

\begin{remark} \label{spr541} \mar{spr541}
The group algebra $L^1_\mathbb C(G)$ is not a $C^*$-algebra. Its
enveloping $C^*$-algebra $C^*(G)$ is called the $C^*$-algebra of a
locally compact group $G$.
\end{remark}

Unitary representations of a locally compact group $G$ and
representations of a group algebra $L^1_\mathbb C(G)$ are related
as follows \cite{dixm}.

\begin{theorem} \label{gns61} \mar{gns61}
There is one-to-one correspondence between the (strongly
continuous) unitary representations $\pi$ of a locally compact
group $G$ and the representations $\pi^L$ (\ref{spr567}) of its
group algebra $L^1_\mathbb C(G)$.
\end{theorem}

\begin{proof}
Let $\pi$ be a (strongly continuous) unitary representation of $G$
in a Hilbert space $E$. Given a bounded positive measure $\m$ on
$G$, let us consider the integrals
\be
\vf_{\xi,\eta}(\m)=\int \lng \pi(g)\xi|\eta\rng \m
\ee
of the coefficient functions $\vf_{\xi,\eta}(g)$ (\ref{spr530})
for all $\xi,\eta \in E$. There exists a bounded operator
$\pi(\m)\in B(E)$ in $E$ such that
\be
\lng \pi(\m)\xi|\eta\rng=\vf_{\xi,\eta}(\m), \qquad  \xi,\eta \in
E.
\ee
It is called the operator-valued integral of $\pi(g)$ with respect
to the measure $\m$, and is denoted by
\mar{qm340}\beq
\pi(\m)=\int \pi(g)\m(g). \label{qm340}
\eeq
The assignment $\m\to\pi(\m)$ provides a representation of a
Banach $^*$-algebra $M^1(G,\mathbb C)$ in $E$. Its restriction
\mar{spr567}\beq
\pi^L(f)=\int \pi(g)f(g)dg\in B(E) \label{spr567}
\eeq
to $L^1_\mathbb C(G)$ is non-degenerate. One says that the
representations (\ref{qm340}) of $M^1(G,\mathbb C)$ and
(\ref{spr567}) of $L^1_\mathbb C(G)$ are determined by a unitary
representation $\pi$ of $G$. Conversely, let $\pi^L$ be a
representation of a Banach $^*$-algebra $L^1_\mathbb C(G)$ in a
Hilbert space $E$. There is a monomorphism $g\to\ve_g$ of a group
$G$ onto a subgroup of Dirac measures $\ve_g$, $g\in G$, of an
algebra $M^1(G,\mathbb C)$. Let $\{u_\iota(q)\}_{\iota\in I}$ be
an approximate identity in $L^1_\mathbb C(G)$. Then
$\{\pi^L(u_\iota)\}$ converges to an element of $B(E)$ which can
be seen as a representation $\pi^L(\ve_\bb)$ of the unit element
$\ve_\bb$ of $M^1(G,\mathbb C)$. Accordingly,
$\{\pi^L(\g(g)u_\iota)\}$ converges to $\pi^L(\ve_g)$. Thereby, we
obtain the (strongly continuous) unitary representation
$\pi(g)=\pi^L(\ve_g)$ of a group $G$ in a Hilbert space $E$.
Moreover, the representation (\ref{spr567}) of $L^1_\mathbb C(G)$
determined by this representation $\pi$ of $G$ coincides with the
original representation $\pi^L$ of $L^1_\mathbb C(G)$.
\end{proof}

Moreover, $\pi$ and $\pi^L$ have the same cyclic vectors and
closed invariant subspaces. In particular, a representation
$\pi^L$ of $L^1_\mathbb C(G)$ is topologically irreducible iff the
associated representation $\pi$ of $G$ is so. It should be
emphasized that, since $L^1_\mathbb C(G)$ is not a $C^*$-algebra,
its topologically irreducible representations need not be
algebraically irreducible. By irreducible representations of a
group $G$, we will mean only its topologically irreducible
representations.

Theorem \ref{gns61} enables us to apply the GNS construction
(Theorem \ref{spr471}) in order to characterize unitary
representations of $G$ by means of positive continuous forms on
$L^1_\mathbb C(G)$.

In accordance with Remark \ref{spr520}, a continuous form on a
group algebra $L_\mathbb C^1(G)$ is defined as
\mar{qm341}\beq
\f(f)=\int \psi(g)f(g)dg   \label{qm341}
\eeq
by an element $\psi$ of a space $L^\infty_\mathbb C(G)$ of
infinite integrable complex functions on $G$ (Remark
\ref{spr520}). However, a function $\psi$ should satisfy the
following additional condition in order that the form
(\ref{qm341}) to be positive.

A continuous complex function $\psi$ on $G$ is called
positive-definite if
\be
\op\sum_{i,j}\psi(g_j^{-1}g_i)\ol\la_i\la_j\geq 0
\ee
for any finite set $g_1,\ldots,g_m$ of elements of $G$ and any
complex numbers $\la_1,\ldots,\la_m$. In particular, if $m=2$ and
$g_1=\bb$, we obtain
\be
\psi(g^{-1})=\ol{\psi(g)}, \qquad \nm\psi(g)\leq\psi(\bb), \qquad
  g\in G,
\ee
i.e., $\psi(\bb)$ is bounded.

\begin{lemma} \label{1080} \mar{1080}
The continuous form (\ref{qm341}) on $L^1_\mathbb C(G)$ is
positive iff $\psi\in L^\infty_\mathbb C(G)$ locally almost
everywhere equals a continuous positive-definite function.
\end{lemma}

Then cyclic representations of a group algebra $L^1_\mathbb C(G)$
and the unitary cyclic representations of a locally compact group
$G$ are defined by continuous positive-definite functions on $G$
in accordance with the following theorem.

\begin{theorem} \label{qm539} \mar{qm539}
Let $\pi_\psi$ be a representation of $L^1_\mathbb C(G)$ in a
Hilbert space $E_\psi$ and $\thh_\psi$ a cyclic vector for
$\pi_\psi$ which are determined by the form (\ref{qm341}). Then
the associated unitary representation $\pi_\psi$ of $G$ in
$E_\psi$ is characterized by a relation
\mar{qm342}\beq
\psi(g)=\lng \pi(g)\thh_\psi|\thh_\psi\rng. \label{qm342}
\eeq
Conversely, a complex function $\psi$ on $G$ is continuous
positive-definite iff there exists a unitary representation
$\pi_\psi$ of $G$ and a cyclic vector $\thh_\psi$ for $\pi_\psi$
such that the equality (\ref{qm342}) holds.
\end{theorem}

By analogy with a Higgs vacuum, one can think of functions $\psi$
in Theorem \ref{qm539} as being classical vacuum fields.

\begin{example} \label{x10} \mar{x10}
Let a group $G$ acts on a Hausdorff topological space $Z$ on the
left. Let $\m$ be a quasi-invariant measure on $Z$ under a
transformation group $G$, i.e., $\g(g)\m=h_g\m$ where $h_g$ is the
Radon--Nikodym derivative in Theorem \ref{spr510}. Then there is a
representation
\mar{x11}\beq
G\ni g: f\to \Pi(g)f, \qquad (\Pi(g)f)(z)=h_g^{1/2}(z)f(gz)
\label{x11}
\eeq
of $G$ in a Hilbert space $L_\mathbb C^2(Z,\m)$ of square
$\m$-integrable complex functions on $Z$ \cite{dao}. It is a
unitary representation due to the equality
\be
\|f\|_\m=\int |f(z)|\m(z)=\int |f(g(z)|^2 \m(g(z))=  \int
h_g(z)|f(g(z)|^2 \m(z)=\|\Pi(g)f\|^2_\m.
\ee
A group $G$ can be equipped with the coarsest topology such that
the representation (\ref{x11}) is strongly continuous. For
instance, let $Z=G$ be a locally compact group, and let $\m=dg$ be
a left Haar measure. Then the representation (\ref{x11}) comes to
the left-regular representation
\mar{spr570}\beq
(\Pi(g)f)(q)=f(g^{-1}q), \qquad f\in L^2_\mathbb C(G), \qquad q\in
G, \label{spr570}
\eeq
of $G$ in a Hilbert space $L^2_\mathbb C(G)$ of square integrable
complex functions on $G$. Note that the above mentioned coarsest
topology on $G$ is coarser then the original one, i.e., the
representation (\ref{spr570}) is strongly continuous.
\end{example}

Let us consider unitary representations of a locally compact group
$G$ which are contained in its left-regular representation
(\ref{spr570}). In accordance with the expression (\ref{qm340}),
the corresponding representation $\Pi(h)$ of a group algebra
$L^1_\mathbb C(G)$ in $L^2_\mathbb C(G)$ reads
\mar{qm344}\beq
(\Pi(h)f)(q)=\int h(g)(\Pi(g)f)(q)dg=\int
h(g)f(g^{-1}q)dg=(h*f)(q). \label{qm344}
\eeq

Let $G$ be a unimodular group. There is the following criterion
that its unitary representation is contained in the left-regular
one.

\begin{theorem} \label{qm360} \mar{qm360}
If a continuous positive-definite function $\psi$ on a unimodular
locally compact group $G$ is square integrable, then the
representation $\pi_\psi$ of $G$ determined by $\psi$ is contained
in the left-regular representation    $\Pi$ (\ref{spr570}) of $G$.
Conversely, let $\pi$ be a cyclic unitary representation of $G$
which is contained in $\Pi$, and let $\thh$ be a cyclic vector for
$\pi$. Then a continuous positive-definite function
$\lng\pi(g)\thh|\thh\rng$ on $G$ is square integrable.
\end{theorem}

The representation $\pi_\psi$ in Theorem \ref{qm360} is
constructed as follows. Given a square integrable continuous
positive-definite function $\psi$ on $G$, there exists a
positive-definite function $\thh\in L^2_\mathbb C(G)$ such that
\be
\psi=\thh*\thh=\thh*\thh^*=\thh^**\ol\thh^*.
\ee
This is a cyclic vector for $\pi_\psi$. The coefficients
(\ref{spr530}) of a representation $\pi_\psi$ read
\be
\vf_{\xi,\eta}(g)=\ol{(\eta*\xi^*)(g)}.
\ee

In particular, if representations $\pi_\psi$ and $\pi_{\psi'}$,
determined by square integrable continuous positive-definite
functions $\psi$ and $\psi'$ on $G$, are irreducible and
inequivalent, then the corresponding cyclic vectors $\thh_\psi$
and $\thh_{\psi'}$ are orthogonal in $L^2_\mathbb C(G)$, while the
functions $\psi$ and $\psi'$ fulfil the relations
\mar{qm345}\beq
\int\psi'(g)\ol\psi(g)dg=0, \qquad \psi*\psi'=0.\label{qm345}
\eeq

Now, let $G$ be a connected locally compact (i.e.,
finite-dimensional) real Lie group. Any unitary representation of
$G$ yields a representation of its right Lie algebra $\cG$ as
follows. In particular, a finite-dimensional unitary
representation of $G$ in a Hilbert space $E$ is analytic, and a
Lie algebra $\cG$ is represented by bounded operators in $E$.

If $\pi$ is an infinite-dimensional (strongly continuous) unitary
representation of $G$ in a Hilbert space $E$, a representation of
a Lie algebra $\cG$ fails to be defined everywhere on $E$ in
general. To construct a carrier space of $\cG$, let us consider a
space $\cK^\infty(G,\mathbb C)\subset L^1_\mathbb C(G)$ of smooth
complex functions on $G$ of compact support and  the vectors
\mar{spr586}\beq
e_f=\pi^L(f)e=\int \pi(g)f(g)edg, \qquad e\in E, \qquad f\in
\cK^\infty(G,\mathbb C), \label{spr586}
\eeq
where $\pi^L$ is the representation (\ref{spr567}) of a group
algebra $L^1_\mathbb C(G)$ \cite{hurt}. The vectors $e_f$
(\ref{spr586}) exemplify smooth vectors of the representation
$\pi$  because, for any $\eta\in E$, the coefficients
$\varphi_{e_f,\eta}(g)$ of $\pi$ are smooth functions on $G$. The
vectors $e_f$ (\ref{spr586}) for all $e\in E$ and $f\in
\cK^\infty(G,\mathbb C)$ constitute a dense vector subspace
$E_\infty$ of $E$. Let $u_a$ be a right-invariant vector field on
$G$ corresponding to an element $a\in \cG$. Then the assignment
\be
\pi_\infty(a): e_f\to \pi^L(u_a\rfloor df)e
\ee
provides a representation of a Lie algebra $\cG$ in $E_\infty$.

\section{GNS construction. Unbounded operators}

There are quantum algebras (e.g., of quantum fields) whose
representations in Hilbert spaces need not be normed. Therefore,
generalizations of the conventional GNS representation of
$C^*$-algebras (Theorem \ref{spr471}) to some classes of unnormed
topological $^*$-algebras has been studied
\cite{book05,hor,schmu}.

In a general setting, by an operator in a Hilbert (or Banach)
space $E$ is meant a linear morphism $a$ of a dense subspace
$D(a)$ of $E$ to $E$. The $D(a)$ is called the domain of an
operator $a$. One says that an operator $b$ on $D(b)$ is an
extension of an operator $a$ on $D(a)$ if $D(a)\subset D(b)$ and
$b|_{D(a)}=a$. For the sake of brevity, we will write $a\subset
b$. An operator $a$ is said to be
 bounded  on $D(a)$ if there exists a real number $r$ such that
\be
\|ae\|\leq r\|e\|, \qquad  e\in D(a).
\ee
If otherwise, it is called unbounded.  Any bounded operator on a
domain $D(a)$ is uniquely extended to a bounded operator
everywhere on $E$.

An operator $a$ on a domain $D(a)$ is called closed if the
condition that a sequence $\{e_i\}\subset D(a)$ converges to $e\in
E$ and that the sequence $\{ae_i\}$ does to $e'\in E$ implies that
$e\in D(a)$ and $e'=ae$. Of course, any operator defined
everywhere on $E$ is closed. An operator $a$ on a domain $D(a)$ is
called closable if it can be extended to a closed operator. The
closure of a closable operator $a$ is defined as the minimal
closed extension of $a$.

Operators $a$ and $b$ in $E$ are called adjoint if
\be
\lng ae|e'\rng=\lng e|be'\rng, \quad  e\in D(a), \quad e'\in D(b).
\ee
Any operator $a$ has a maximal adjoint operator  $a^*$, which is
closed. Of course, $a\subset a^{**}$ and $b^*\subset a^*$ if
$a\subset b$. An operator $a$ is called  symmetric if it is
adjoint to itself, i.e., $a\subset a^*$. Hence, a symmetric
operator is closable. One can obtain the following chain of
extensions of a symmetric operator:
\be
a\subset\ol a\subset a^{**}\subset a^*=\ol a^*=a^{***}.
\ee
In particular, if $a$ is a symmetric operator, so are $\ol a$ and
$a^{**}$. At the same time, the maximal adjoint operator $a^*$ of
a symmetric operator $a$ need not be symmetric. A symmetric
operator $a$ is called  self-adjoint if $a=a^*$, and it is called
 essentially self-adjoint if $\ol a=a^*=\ol a^*$. It should be
emphasized that a symmetric operator $a$ is sometimes called
essentially self-adjoint if $a^{**}=a^*$. We here follow the
terminology of \cite{pow1}. If $a$ is a closed operator, the both
notions coincide. For bounded operators, the notions of symmetric,
self-adjoint and essentially self-adjoint operators coincide.

Let $E$ be a Hilbert space. A pair $(B,D)$ of a dense subspace $D$
of $E$ and a unital algebra $B$ of (unbounded) operators in $E$ is
called the  $Op^*$-algebra ($O^*$-algebra in the terminology of
\cite{schmu}) on a domain $D$ if, whenever $b\in B$, we have
\cite{hor,pow1}: (i) $D(b)=D$ and $bD\subset D$, (ii) $D\subset
D(b^*)$, (iii) $b^*|_D\subset B$. An algebra $B$ is provided with
the involution $b\to b^+=b^*|_D$, and its elements are closable.

A representation $\pi(A)$ of a $^*$-algebra $A$ in a Hilbert space
$E$ is defined as a homomorphism of $A$ to an $Op^*$-algebra
$(B,D(\pi))$ of (unbounded) operators in $E$ such that
$D(\pi)=D(\pi(a))$ for all $a\in A$ and this representation is
Hermitian, i.e., $\pi(a^*)\subset \pi(a)^*$ for all $a\in A$. In
this case, one also considers the representations
\be
&& \ol\pi: a \to \ol\pi(a)=\ol{\pi(a)}|_{D(\ol\pi)},
\qquad D(\ol\pi)=\op\bigcap_{a\in A} D(\ol{\pi(a)}),\\
&&  \pi^*: a \to \pi^*(a)=\pi(a^*)^*|_{D(\pi^*)},
\qquad D(\pi^*)=\op\bigcap_{a\in A} D(\pi(a)^*),
\ee
called the closure  of a representation $\pi$, and the adjoint
representation, respectively. There are the representation
extensions $\pi\subset\ol\pi\subset\pi^*$, where $\pi_1\subset
\pi_2$ means $D(\pi_1)\subset D(\pi_2)$. The representations
$\ol\pi$ and $\pi^{**}$ are Hermitian, while $\pi^*=\ol\pi^*$. A
Hermitian representation $\pi(A)$ is said to be  closed if
$\pi=\ol\pi$, and it is self-adjoint if $\pi=\pi^*$. Herewith, a
representation $\pi(A)$ is closed (resp. self-adjoint) if one of
operators $\pi(A)$ is closed (resp. self-adjoint).

A representation domain $D(\pi)$ is endowed with the
graph-topology.  It is generated by neighborhoods of the origin
\be
U(M,\ve)=\{x\in D(\pi)\,:\,\op\sum_{a\in M} \|\pi(a)x\|<\ve\},
\ee
where $M$ is a finite subset of elements of $A$. All operators of
$\pi(A)$ are continuous with respect to this topology. Let us note
that the graph-topology is finer than the relative topology on
$D(\pi)\subset E$, unless all operators $\pi(a)$, $a\in A$, are
bounded \cite{schmu}.

Let $\ol N^g$ denote the closure of a subset $N\subset D(\pi)$
with respect to the graph-topology. An element $\thh\in D(\pi)$ is
called strongly cyclic  (cyclic in the terminology of
\cite{schmu}) if
\be
D(\pi)\subset \ol{(\pi(A)\thh)}^g.
\ee
Then the GNS representation Theorem \ref{spr471} can be
generalized to Theorem \ref{gns} \cite{hor,schmu}.

Similarly to Remark \ref{gns123}, we say that representations
$\pi_1$ and $\pi_2$ of a $^*$-algebra $A$ are equivalent if there
exists an isomorphism $\g$ of their carrier spaces such that
\be
D(\pi_1)=\g(D(\pi_2)), \qquad \pi_1(a)=\g\circ\pi_2(a)\circ
\g^{-1}, \qquad a\in A.
\ee
In particular, if representations are equivalent, their kernels
coincide with each other.

Accordingly, states $f$ and $f'$ of a unital topological
$^*$-algebra $A$ in Theorem \ref{gns} are called equivalent if
they define equivalent representations $\pi_f$ and$\pi_{f'}$.

By analogy with Theorem \ref{spr474}, one can show the following.

\begin{theorem} \mar{spr474'} \label{spr474'}
Positive continuous forms $f$ and $f'$ on a unital topological
$^*$-algebra $A$ are equivalent if there exist elements $b,b'\in
A$ such that
\be
f'(a)=f(b^+ab), \qquad f(a)=f'(b'^+ab'), \qquad a\in A.
\ee
\end{theorem}

We point out the particular class of nuclear barreled
$^*$-algebras. Let $A$ be a locally convex topological
$^*$-algebra whose topology is defined by a set of multiplicative
seminorms $p_\iota$ which satisfy the condition
\be
p_\iota(a^*a)=p_\iota(a)^2, \qquad a\in A.
\ee
It is called a $b^*$-algebra.  A unital $b^*$-algebra as like as a
$C^*$-algebra is regular and symmetric, i.e., any element $(\bb
+a^*a)$, $a\in A$, is invertible and, moreover, $(\bb +a^*a)^{-1}$
is bounded \cite{all,igu}.
  The
$b^*$-algebras are related to $C^*$-algebras as follows.

\begin{theorem}
Any $b^*$-algebra is the Hausdorff projective limit of a family of
$C^*$-algebras, and {\it vice versa} \cite{igu}.
\end{theorem}

In particular, every $C^*$-algebra $A$ is a  barreled
$b^*$-algebra,  i.e., every absorbing balanced closed subset is a
neighborhood of the origin of $A$.

Let us additionally assume that $A$ is a nuclear algebra, i.e., a
nuclear space (Section 14.2). Then we have the following variant
of the GNS representation Theorem \ref{gns} \cite{igu}

\begin{theorem} \mar{gns36} \label{gns36}
Let $A$ be a unital nuclear barreled $b^*$-algebra and $f$ a
positive continuous form on $A$. There exists a unique (up to
unitary equivalence) cyclic representation $\pi_f$ of $A$ in a
Hilbert space $E_f$ by operators on a common invariant domain
$D\subset E_f$. This domain can be topologized to conform a rigged
Hilbert space such that all the operators representing $A$ are
continuous on $D$.
\end{theorem}

\begin{example} \mar{gns90} \label{gns90}
The following is an example of nuclear barreled $b^*$-algebras
which is very familiar from quantum field theory
\cite{sard91,ccr,axiom}. Let $Q$ be a nuclear space (Section
14.2). Let us consider a direct limit
\mar{x2}\beq
\wh\ot Q =\mathbb C \oplus Q \oplus Q\wh\ot Q\oplus \cdots\oplus
Q^{\wh\ot n}\oplus\cdots \label{x2}
\eeq
of vector spaces
\mar{x2a}\beq
 \wh\ot^{\leq n} Q =\mathbb C \oplus
Q \oplus Q\wh\ot Q\oplus\cdots \oplus Q^{\wh\ot n}, \label{x2a}
\eeq
where $\wh\ot$ is the topological tensor product with respect to
Grothendieck's topology (which coincides with the $\ve$-topology
on the tensor product of nuclear spaces \cite{piet}). The space
(\ref{x2}) is provided with the inductive limit topology,  the
finest topology such that the morphisms $\wh\ot^{\leq n} Q\to
\wh\ot Q$ are continuous and, moreover, are imbeddings
\cite{trev}. A convex subset $V$ of $\wh\ot Q$ is a neighborhood
of the origin in this topology iff $V\cap \wh\ot^{\leq n} Q$ is so
in $\wh\ot^{\leq n} Q$. Furthermore, one can show that $\wh\ot Q$
is a unital nuclear barreled LF-algebra with respect to a tensor
product \cite{belang}. The  LF-property implies that a linear form
$f$ on $\wh\ot Q$ is continuous iff the restriction of $f$ to each
$\wh\ot^{\leq n} Q$ is so \cite{trev}. If a continuous conjugation
$*$ is defined on $Q$, the algebra $\wh\ot Q$ is involutive with
respect to the operation
\mar{x3}\beq
*(q_1\ot\cdots \ot q_n)=q_n^*\ot\cdots \ot q_1^* \label{x3}
\eeq
on $Q^{\ot n}$ extended by continuity and linearity to $Q^{\wh\ot
n}$. Moreover, $\wh\ot Q$ is a $b^*$-algebra as follows. Since $Q$
is a nuclear space, there is a family $\|.\|_k$, $k\in\mathbb
N_+$, of continuous norms on $Q$. Let $Q_k$ denote the completion
of $Q$ with respect to the norm $\|.\|_k$. Then one can show that
the tensor algebra $\ot Q_k$ is a $C^*$-algebra and that $\wh\ot
Q$ (\ref{x2}) is a projective limit of these $C^*$-algebras with
respect to morphisms $\ot Q_{k+1}\to \ot Q_k$ \cite{igu}. Since
$\wh\ot Q$ (\ref{x2}) is a nuclear barreled $b^*$-algebra, it
obeys GNS representation Theorem \ref{gns36}. Herewith, let us
note that, due to the LF-property, a positive continuous form $f$
on $\wh\ot Q$ is defined by a family of its restrictions $f_n$ to
tensor products $\wh\ot^{\leq n} Q$. One also can restrict a form
$f$ and the corresponding representation $\pi_f$ to a tensor
algebra
\mar{gns100}\beq
A_Q=\ot Q\subset \wh\ot Q \label{gns100}
\eeq
of $Q$. However, a representation $\pi_f(A_Q)$ of $A_Q$ in a
Hilbert space $E_f$ need not be cyclic.
\end{example}

In quantum field theory, one usually choose $Q$ the Schwartz space
of functions of rapid decrease (Sections 9 and 10).

\section{Example. Commutative nuclear groups}

Following Example \ref{gns90} in a case of a real nuclear space
$Q$, let us consider a commutative tensor algebra
\mar{gns101}\beq
B_Q=\mathbb R \oplus Q\oplus Q\vee Q\oplus\cdots\oplus \op\vee^n
Q\oplus \cdots. \label{gns101}
\eeq
Provided with the direct sum topology, $B_Q$ becomes a unital
topological $^*$-algebra. It coincides with the universal
enveloping algebra of the Lie algebra $T_Q$ of an additive Lie
group $T(Q)$ of translations in $Q$. Therefore, one can obtain the
states of an algebra $B_Q$ by constructing cyclic strongly
continuous unitary representations of a nuclear Abelian group
$T(Q)$.

\begin{remark} Let us note that, in contrast to that studied in
Section 5, a nuclear group $T(Q)$ is not locally compact, unless
$Q$ is finite-dimensional.
\end{remark}

A cyclic strongly continuous unitary representation $\pi$ of
$T(Q)$ in a Hilbert space $(E,\lng.|.\rng_E)$ with a normalized
cyclic vector $\theta\in E$ yields a complex function
\be
Z(q)=\lng \pi(T(q))\theta|\theta\rng_E
\ee
on $Q$. This function is proved to be continuous and
positive-definite, i.e., $Z(0)=1$ and
\be
\op\sum_{i,j} Z(q_i-q_j)\ol\la_i\la_j\geq 0
\ee
for any finite set $q_1,\ldots,q_m$ of elements of $Q$ and
arbitrary complex numbers $\la_1,\ldots,\la_m$.

In accordance with the well-known Bochner theorem for nuclear
spaces (Theorem \ref{spr525}), any continuous positive-definite
function $Z(q)$ on a nuclear space $Q$ is the Fourier transform
\mar{qm545}\beq
Z(q)=\int\exp[i\lng q,u\rng]\m(u) \label{qm545}
\eeq
of a positive measure $\m$ of total mass 1 on the dual $Q'$ of
$Q$. Then the above mentioned representation $\pi$ of $T(Q)$ can
be given by the operators
\mar{q2}\beq
T_Z(q)\rho(u)=\exp[i\langle q,u\rangle]\rho(u)  \label{q2}
\eeq
in a Hilbert space $L_\mathbb C^2(Q',\m)$ of equivalence classes
of square $\m$-integrable complex functions $\rho(u)$ on $Q'$. A
cyclic vector $\thh$ of this representation is the
$\m$-equivalence class $\thh\ap_\m 1$ of a constant function
$\rho(u)=1$. Then we have
\mar{ccr1}\beq
Z(q)=\lng T_Z(q)\thh|\thh\rng_\m=\int \exp[i\langle q,u\rangle]\m.
\label{ccr1}
\eeq

Conversely, every positive measure $\m$ of total mass 1 on the
dual $Q'$ of $Q$ defines the cyclic strongly continuous unitary
representation (\ref{q2}) of a group $T(Q)$. By virtue of the
above mentioned Bochner theorem, it follows that every continuous
positive-definite function $Z(q)$ on $Q$ characterizes a cyclic
strongly continuous unitary representation (\ref{q2}) of a nuclear
Abelian group $T(Q)$. We agree to call $Z(q)$ the generating
function of this representation.

\begin{remark}
The representation (\ref{q2}) need not be topologically
irreducible. For instance, let $\rho(u)$ be a function on $Q'$
such that a set where it vanishes is not a $\m$-null subset of
$Q'$. Then the closure of a set $T_Z(Q)\rho$ is a $T(Q)$-invariant
closed subspace of $L_\mathbb C^2(Q',\m)$.
\end{remark}

Different generating functions $Z(q)$ determine inequivalent
representations $T_Z$ (\ref{q2}) of $T(Q)$ in general. One can
show the  following \cite{gelf64}.

\begin{theorem} \label{gns77} \mar{gns77}
Distinct generating functions $Z(q)$ and $Z'(q)$ determine
equivalent representations $T_Z$ and $T_{Z'}$ (\ref{q2}) of $T(Q)$
in Hilbert spaces $L^2_\mathbb C(Q',\m)$ and $L^2_\mathbb
C(Q',\m')$  iff they are the Fourier transform of equivalent
measures on $Q'$.
\end{theorem}

Indeed, let
\mar{qm581}\beq
\m'= s^2\m, \label{qm581}
\eeq
where a function $s(u)$ is strictly positive almost everywhere on
$Q'$, and $\m(s^2)=1$. Then the map
\mar{qm580}\beq
L^2_\mathbb C(Q',\m')\ni\rho(u)\to s(u)\rho(u)\in L^2_\mathbb C
(Q',\m) \label{qm580}
\eeq
provides an isomorphism between the representations $T_{Z'}$ and
$T_Z$.

Similarly to the case of finite-dimensional Lie groups (Section
4), any strongly continuous unitary representation (\ref{q2}) of a
nuclear group $T(Q)$ implies a representation of its Lie algebra
by operators
\mar{gns102}\beq
\f(q)\rho(u)=\lng q,u\rng\rho(u) \label{gns102}
\eeq
in the same Hilbert space $L^2_\mathbb C(Q',\m)$. Their mean
values read
\mar{gns103}\beq
\lng\f(q)\rng=\om_\thh(\f(q))=\lng\f(q)\rng=\int \lng q,u\rng
\mu(u). \label{gns103}
\eeq

The representation (\ref{gns102}) is extended to that of the
universal enveloping algebra $B_Q$ (\ref{gns101}).

\begin{remark} \label{gns109} \mar{gns109}
Let us consider representations of $T(Q)$  with generating
functions $Z(q)$ such that $\mathbb R\ni t\to Z(tq)$ is an
analytic function on $\mathbb R$ at $t=0$ for all $q\in Q$. Then
one can show that a function $\lng q|u\rng$ on $Q'$ is square
$\m$-integrable for all $q\in Q$ and that, consequently, the
operators $\f(q)$ (\ref{gns102}) are bounded everywhere in a
Hilbert space $L^2_\mathbb C(Q',\m)$. Moreover, the corresponding
mean values of elements of $B_Q$ can be computed by the formula
\mar{ccr8}\beq
\lng \f(q_1)\cdots\f(q_n)\rng=i^{-n} \frac{\dr}{\dr\al^1}
\cdots\frac{\dr}{\dr\al^n}Z(\al^iq_i)|_{\al^i=0}=  \int\lng
q_1,u\rng\cdots\lng q_n, u \rng \mu(u). \label{ccr8}
\eeq
\end{remark}

\section{Infinite canonical commutation relations}

The canonical commutation relations (CCR) are of central
importance in AQT as a method of canonical quantization. A
remarkable result about CCR for finite degrees of freedom is the
Stone--von Neumann uniqueness theorem which states that all
irreducible representations of these CCR are unitarily equivalent
\cite{petz}. On the contrary, CCR of infinite degrees of freedom
admit infinitely many inequivalent irreducible representations
\cite{flor}.

One can provide the comprehensive description of representations
of CCR modelled over an infinite-dimensional nuclear space $Q$
\cite{gelf64,book05,ccr}.

Let $Q$ be a real nuclear space endowed with a non-degenerate
separately continuous Hermitian form $\lng.|.\rng$. This Hermitian
form brings $Q$ into a Hausdorff pre-Hilbert space. A nuclear
space $Q$, the completion $\wt Q$ of a pre-Hilbert space $Q$, and
the dual $Q'$ of $Q$ make up the rigged Hilbert space $Q\subset
\wt Q\subset Q'$ (\ref{spr450}).

Let us consider a group $G(Q)$ of triples $g=(q_1,q_2,\la)$ of
elements $q_1$, $q_2$ of $Q$ and complex numbers $\la$, $|\la|=1$,
which are subject to multiplications
\mar{qm541}\beq
(q_1,q_2,\la)(q'_1,q'_2,\la')=(q_1+q'_1,q_2+q'_2, \exp[i\lng
q_2,q'_1\rng] \la\la'). \label{qm541}
\eeq
It is a Lie group whose group space is a nuclear manifold modelled
over a vector space $W=Q\oplus Q\oplus \mathbb R$. Let us denote
$T(q)=(q,0,0)$,  $P(q)=(0,q,0)$. Then the multiplication law
(\ref{qm541}) takes a form
\mar{qm543}\beq
T(q)T(q')=T(q+q'), \quad P(q)P(q')=P(q+q'), \quad
P(q)T(q')=\exp[i\lng q|q'\rng]T(q')P(q). \label{qm543}
\eeq
Written in this form, $G(Q)$ is called the nuclear Weyl CCR group.

The complexified Lie algebra of a nuclear Lie group $G(Q)$ is the
unital Heisenberg CCR algebra $\ccG(Q)$. It is generated by
elements $I$, $\f(q)$, $\pi(q)$, $q\in Q$, which obey the
Heisenberg CCR
\mar{qm540}\beq
[\f(q),I]=\pi(q),I]=[\f(q),\f(q')]=[\pi(q),\pi(q')]=0, \qquad
[\pi(q),\f(q')]=-i\lng q|q'\rng I. \label{qm540}
\eeq
There is the exponential map
\be
T(q)=\exp[i\f(q)], \qquad P(q)=\exp[i\pi(q)].
\ee

Due to the relation (\ref{1085}), the normed topology on a
pre-Hilbert space $Q$ defined by a Hermitian form $\lng.|.\rng$ is
coarser than the nuclear space topology. The latter is metric,
separable and, consequently, second-countable. Hence, a
pre-Hilbert space $Q$ also is second-countable and, therefore,
admits a countable orthonormal basis. Given such a basis $\{q_i\}$
for $Q$, the Heisenberg CCR (\ref{qm540}) take a form
\be
[\f(q_j),I]=\pi(q_j),I]=[\f(q_j),\f(q_k)]=[\pi(q_k),\pi(q_j)]=0,
\qquad [\pi(q_j),\f(q_k)]=-i\dl_{jk}I.
\ee

A Weyl CCR group $G(Q)$ contains two nuclear Abelian subgroups
$T(Q)$ and $P(Q)$. Following the representation algorithm in
\cite{gelf64}, we first construct representations of a nuclear
Abelian group $T(Q)$.  Then these representations under certain
conditions can be extended to representations of a Weyl CCR group
$G(Q)$.

Following Section 7, we treat a nuclear Abelian group $T(Q)$ as
being a group of translations in a nuclear space $Q$. Let us
consider its cyclic strongly continuous unitary representation
$T_Z$ (\ref{q2}) in a Hilbert space $L^2_\mathbb C(Q',\m)$ of
equivalence classes of square $\m$-integrable complex functions
$\rho(u)$ on the dual $Q'$ of $Q$ which is defined by the
generating function $Z$ (\ref{qm545}). This representation can be
extended to a Weyl CCR group $G(Q)$ if a measure $\m$ possesses
the following property.

Let $u_q$, $q\in Q$, be an element of $Q'$ given by the condition
$\lng q',u_q\rng=\lng q'|q\rng$, $q'\in Q$. These elements form
the range of a monomorphism $Q\to Q'$ determined by a Hermitian
form $\lng.|.\rng$ on $Q$. Let a measure $\m$ in the expression
(\ref{qm545}) remain equivalent under translations
\be
Q'\ni u\to u+u_q \in Q', \qquad   u_q\in Q\subset Q',
\ee
in $Q'$, i.e.,
\mar{qm547}\beq
\m(u+u_q)=a^2(q,u)\m(u), \qquad  u_q\in Q\subset Q', \label{qm547}
\eeq
where a function $a(q,u)$ is square $\m$-integrable and strictly
positive almost everywhere on $Q'$. This function fulfils the
relations
\mar{qm555}\beq
a(0,u)=1, \qquad a(q+q',u)=a(q,u)a(q',u+u_q). \label{qm555}
\eeq
A measure on $Q'$ obeying the condition (\ref{qm547}) is called
translationally quasi-invariant, though it does not remain
equivalent under an arbitrary translation in $Q'$, unless $Q$ is
finite-dimensional.

Let the generating function $Z$ (\ref{qm545}) of a cyclic strongly
continuous unitary representation of a nuclear group $T(Q)$ be the
Fourier transform  of a translationally quasi-invariant measure
$\m$ on $Q'$.  Then one can extend the representation (\ref{q2})
of this group to a unitary strongly continuous representation of a
Weyl CCR group $G(Q)$ in a Hilbert space $L^2_\mathbb C(Q',\m)$ by
the operators (\ref{x11}) in Example \ref{x10}. These operators
read
\mar{qm548}\beq
P_Z(q)\rho(u)=a(q,u)\rho(u+u_q). \label{qm548}
\eeq
Herewith, the following is true.

\begin{theorem} \label{gns78} \mar{gns78}
Equivalent representations of a group $T(Q)$ are extended to
equivalent representations of a Weyl CCR group $G(Q)$.
\end{theorem}

\begin{proof}
Let $\m'$ (\ref{qm581}) be a $\m$-equivalent positive measure of
total mass 1 on $Q'$. The equality
\be
\m'(u+u_q)=s^{-2}(u)a^2(q,u)s^2(u+u_q)\m'(u)
\ee
shows that it also is translationally quasi-invariant. Then the
isomorphism (\ref{qm580}) between representations $T_Z$ and
$T_{Z'}$ of a nuclear Abelian group $T(Q)$ is extended to the
isomorphism
\be
P_{Z'}(q)= s^{-1}P_Z(q)s: \rho(u)\to
s^{-1}(u)a(q,u)s(u+u_q)\rho(u+u_q)
\ee
of the corresponding representations of a Weyl CCR group $G(Q)$.
\end{proof}

Similarly to the case of finite-dimensional Lie groups (Section
4), any strongly continuous unitary representation $T_Z$
(\ref{q2}), $P_Z$ (\ref{qm548}) of a nuclear Weyl CCR group $G(Q)$
implies a representation of its Lie algebra $\ccG(Q)$ by
(unbounded) operators in the same Hilbert space $L^2_\mathbb
C(Q',\m)$ \cite{book05,ccr}.  This representation reads
\mar{qm549,62}\ben
&& \f(q)\rho(u)=\lng q,u\rng\rho(u), \qquad
    \pi(q)\rho(u)=-i(\dl_q+\eta(q,u))\rho(u), \label{qm549}\\
&& \dl_q\rho(u)=\op\lim_{\al\to 0}\al^{-1}[\rho(u+\al u_q)-\rho(u)],
\qquad \al\in\mathbb R,\nonumber\\
&& \eta(q,u)=\op\lim_{\al\to 0}\al^{-1}[a(\al q,u)-1].\label{qm562}
\een
It follows at once from the relations (\ref{qm555}) that
\be
&& \dl_q\dl_{q'}=\dl_{q'}\dl_q, \qquad
\dl_q(\eta(q',u))=\dl_{q'}(\eta(q,u)),\\
&& \dl_q=-\dl_{-q}, \qquad \dl_q(\lng q',u\rng)=\lng q'|q\rng,\\
&& \eta(0,u)=0, \quad  u\in Q',\qquad
\dl_q\thh=0,  \quad  q\in Q.
\ee
Then it is easily justified that the operators (\ref{qm549})
fulfil the Heisenberg CCR (\ref{qm540}). The unitarity condition
implies the conjugation rule
\be
\lng q,u\rng^*=\lng q,u\rng, \qquad \dl_q^*=-\dl_q -2\eta(q,u).
\ee
Hence, the operators (\ref{qm549}) are Hermitian.

The operators $\pi(q)$ (\ref{qm549}), unlike $\f(q)$, act in a
subspace $E_\infty$ of all smooth complex functions in
$L^2_\mathbb C (Q',\m)$ whose derivatives of any order also
belongs to $L^2_\mathbb C(Q',\m)$. However, $E_\infty$ need not be
dense in a Hilbert space $L^2_\mathbb C(Q',\m)$, unless $Q$ is
finite-dimensional.

A space $E_\infty$ also is a carrier space of a representation of
the universal enveloping algebra $\ol\ccG(Q)$ of a Heisenberg CCR
algebra $\ccG(Q)$. The representations of $\ccG(Q)$ and
$\ol\ccG(Q)$ in $E_\infty$ need not be irreducible. Therefore, let
us consider a subspace $E_\thh=\ol\ccG(Q)\thh$ of $E_\infty$,
where $\thh$ is a cyclic vector for a representation of a Weyl CCR
group in $L^2_\mathbb C(Q',\m)$. Obviously, a representation of a
Heisenberg CCR algebra $\ccG(Q)$ in $E_\thh$ is algebraically
irreducible.

One also introduces creation and annihilation operators
\mar{qm552}\beq
a^\pm(q)=\frac{1}{\sqrt 2}[\f(q)\mp i\pi(q)]=\frac{1}{\sqrt
2}[\mp\dl_q \mp\eta(q,u) + \lng q,u\rng]. \label{qm552}
\eeq
They obey the conjugation rule $(a^\pm(q))^*=a^\mp(q)$ and the
commutation relations
\be
[a^-(q), a^+(q')]=\lng q|q'\rng\bb, \qquad
[a^+(q),a^+(q')]=[a^-(q),a^-(q')]=0.
\ee
The particle number operator $N$ in a carrier space $E_\thh$ is
defined by  conditions
\be
[N,a^\pm(q)]=\pm a^\pm(q)
\ee
up to a summand $\la\bb$. With respect to a countable orthonormal
basis $\{q_k\}$, this operator $N$ is given by a sum
\mar{qm591}\beq
N=\op\sum_k a^+(q_k)a^-(q_k), \label{qm591}
\eeq
but need not be defined everywhere in $E_\thh$, unless $Q$ is
finite-dimensional.

Gaussian measures given by the Fourier transform (\ref{spr523})
exemplify a physically relevant class of translationally
quasi-invariant measures on the dual $Q'$ of a nuclear space $Q$.
Their Fourier transforms obey the analiticity condition in Remark
\ref{gns109}.

Let $\m_K$ denote a Gaussian measure on $Q'$ whose Fourier
transform is a generating function
\mar{qm563}\beq
Z_K=\exp[-\frac12 M_K(q)] \label{qm563}
\eeq
with the covariance form
\mar{qm560}\beq
M_K(q)=\lng K^{-1}q|K^{-1}q\rng, \label{qm560}
\eeq
where $K$ is a bounded invertible operator in the Hilbert
completion $\wt Q$ of $Q$ with respect to a Hermitian form
$\lng.|.\rng$. The Gaussian measure $\m_K$ is translationally
quasi-invariant, i.e.,
\be
\m_K(u+u_q)=a_K^2(q,u)\m_K(u).
\ee
Using the formula (\ref{ccr8}), one can show that
\mar{qm561}\beq
a_K(q,u)= \exp[-\frac14 M_K(Sq)- \frac12\lng Sq,u\rng],
\label{qm561}
\eeq
where $S=KK^*$ is a bounded Hermitian operator in $\wt Q$.

Let us construct a representation of a CCR algebra $\ccG(Q)$
determined by the generating function $Z_K$ (\ref{qm563}).
Substituting the function (\ref{qm561}) into the formula
(\ref{qm562}), we obtain
\be
\eta(q,u)= -\frac12\lng Sq,u\rng.
\ee
Hence, the operators $\f(q)$ and $\pi(q)$ (\ref{qm549}) take a
form
\mar{qm565}\beq
\f(q)=\lng q,u\rng, \qquad \pi(q)=-i(\dl_q-\frac12\lng Sq,u\rng).
\label{qm565}
\eeq
Accordingly, the creation and annihilation operators (\ref{qm552})
read
\mar{qm566}\beq
a^\pm(q)=\frac{1}{\sqrt 2}[\mp\dl_q \pm \frac12\lng Sq,u\rng +
\lng q,u\rng]. \label{qm566}
\eeq
They act on the subspace $E_\thh$, $\thh\ap_{\m_K}1$, of a Hilbert
space $L^2_\mathbb C(Q',\m_K)$, and they are Hermitian with
respect to a Hermitian form $\lng.|.\rng_{\m_K}$ on $L^2_\mathbb
C(Q',\m_K)$.

\begin{remark} \label{ccr10} \mar{ccr10}
If a representation of CCR is characterized by the Gaussian
generating function (\ref{qm563}), it is convenient for a
computation to express all operators into the operators $\dl_q$
and $\f(q)$, which obey commutation relations
\be
[\dl_q,\f(q')]=\lng q'|q\rng.
\ee
For instance, we have
\be
\pi(q)=-i\dl_q -\frac{i}{2}\f(Sq).
\ee
The mean values $\lng\f(q_1)\cdots\f(q_n)\dl_q\rng$ vanish, while
the meanvalues $\lng\f(q_1)\cdots\f(q_n)\rng$, defined by the
formula (\ref{ccr8}), obey the Wick theorem relations
\mar{ccr20}\beq
\lng\f(q_1)\cdots\f(q_n)\rng =\sum \lng\f(q_{i_1})\f(q_{i_2})\rng
\cdots \lng\f(q_{i_{n-1}})\f(q_{i_n})\rng, \label{ccr20}
\eeq
where the sum runs through all partitions of a set $1,\ldots,n$ in
ordered pairs $(i_1<i_2),\ldots(i_{n-1}<i_n)$, and where
\be
\lng\f(q)\f(q')\rng=\lng K^{-1}q|K^{-1}q'\rng.
\ee
\end{remark}

In particular, let us put $K=\sqrt2\cdot\bb$. Then the generating
function (\ref{qm563}) takes a form
\mar{qm567}\beq
Z_{\mathrm F}(q)=\exp[-\frac14\lng q|q\rng], \label{qm567}
\eeq
and defines the Fock representation of a Heisenberg CCR algebra
$\ccG(Q)$:
\mar{gns81}\ben
&& \f(q)=\lng q,u\rng, \qquad  \pi(q)=-i(\dl_q-\lng q,u\rng), \label{gns81}\\
&& a^+(q)=\frac{1}{\sqrt 2}[-\dl_q + 2\lng q,u\rng], \qquad
a^-(q)=\frac{1}{\sqrt 2}\dl_q. \nonumber
\een
Its carrier space is the subspace $E_\thh$, $\thh\ap_{\m_{\mathrm
F}}1$
    of the Hilbert space
$L^2_\mathbb C(Q',\m_{\mathrm F})$, where $\m_{\mathrm F}$ denotes
a Gaussian measure whose Fourier transform is (\ref{qm567}). We
agree to call it the Fock measure.

The Fock representation (\ref{gns81}) up to an equivalence is
characterized by the existence of a cyclic vector $\thh$ such that
\mar{qm570}\beq
a^-(q)\thh=0, \qquad  q\in Q. \label{qm570}
\eeq
For the representation in question, this is $\thh\ap_{\m_{\mathrm
F}}1$. An equivalent condition is that the particle number
operator $N$ (\ref{qm591}) exists and its spectrum is lower
bounded. The corresponding eigenvector of $N$ in $E_\thh$ is
$\thh$ itself so that $N\thh=0$. Therefore, it is treated a
particleless vacuum.

A glance at the expression (\ref{qm566}) shows that the condition
(\ref{qm570}) does not hold, unless $Z_K$ is $Z_{\mathrm F}$
(\ref{qm567}). For instance, the particle number operator in the
representation (\ref{qm566}) reads
\be
&& N=\op\sum_j a^+(q_j)a^-(q_j)= \op\sum_j[-\dl_{q_j}\dl_{q_j}
+S^j_k\lng q_k,u\rng\dr_{q_j} + \\
&& \qquad
(\dl_{km}-\frac14 S^j_kS^j_m)\lng q_k,u\rng\lng q_m,u\rng -
(\dl_{jj}-\frac12 S^j_j)],
\ee
where $\{q_k\}$ is an orthonormal basis for a pre-Hilbert space
$Q$. One can show that this operator is defined everywhere on
$E_\thh$ and is lower bounded  only if the operator $S$ is a sum
of the scalar operator $2\cdot\bb$ and a nuclear operator in $\wt
Q$, in particular, if
\be
\mathrm{Tr}(\bb -\frac12 S)<\infty.
\ee
This condition also is sufficient for measures $\m_K$ and $\m_F$
(and, consequently, the corresponding representations) to be
equivalent \cite{gelf64}. For instance, a generating function
\be
Z_c(q)=\exp[-\frac{c^2}{2}\lng q|q\rng], \qquad c^2\neq \frac12,
\ee
defines a non-Fock representation of nuclear CCR.

\begin{remark} \label{gns79} \mar{gns79}
Since the Fock measure $\m_{\mathrm F}$ on $Q'$ remains equivalent
only under translations by vectors $u_q\in Q\subset Q'$, the
measure
\be
\m_\si(u) =\m_{\mathrm F}(u+\si), \qquad \si\in Q'\setminus Q,
\ee
on $Q'$ determines a non-Fock representation of nuclear CCR.
Indeed, this measure is translationally quasi-invariant:
\be
\m_\si(u+u_q)=a^2_\si(q,u)\m_\si(u),\qquad a_\si(q,u)=a_{\mathrm
F}(q,u-\si),
\ee
and its Fourier transform
\be
Z_\si(q)=\exp[i\lng q,\si\rng]Z_{\mathrm F}(q)
\ee
is a positive-definite continuous function on $Q$. Then the
corresponding representation of a CCR algebra is given by
operators
\mar{qm597}\beq
a^+(q)=\frac{1}{\sqrt 2}(-\dl_q + 2\lng q,u\rng -\lng
q,\si\rng),\qquad a^-(q)=\frac{1}{\sqrt 2}(\dl_q + \lng
q,\si\rng). \label{qm597}
\eeq
In comparison with the all above mentioned representations, these
operators possess non-vanishing vacuum mean values
\be
\lng a^\pm(q)\thh|\thh\rng_{\m_{\mathrm F}}=\lng q,\si\rng.
\ee
If $\si\in Q\subset Q'$, the representation (\ref{qm597}) becomes
equivalent to the Fock representation (\ref{gns81}) due to a
morphism
\be
\rho(u)\to \exp[-\lng \si,u\rng]\rho(u+\si).
\ee
\end{remark}

\begin{remark} \label{gns80} \mar{gns80}
Let us note that the non-Fock representation (\ref{qm565}) of the
CCR algebra (\ref{qm540}) in a Hilbert space $L^2_\mathbb
C(Q',\m_K)$ is the Fock representation
\be
\f_K(q)=\f(q)=\lng q,u\rng,\qquad \pi_K(q)=\pi(S^{-1}q) =
-i(\dl^K_q -\frac12\lng q,u\rng), \qquad \dl^K_q =\dl_{S^{-1}q},
\ee
of a CCR algebra $\{\f_K(q),\pi_K(q),I\}$, where
\be
[\f_K(q),\pi_K(q)]=i\lng K^{-1}q|K^{-1}q'\rng I.
\ee
\end{remark}

\section{Free quantum fields}

There are two main algebraic formulation of QFT. In the framework
of the first one, called local QFT, one associates to a certain
class of subsets of a Minkowski space a net of von Neumann, $C^*$-
or $Op^*$-algebras which obey certain axioms
\cite{araki,buch,haag,halv,hor}. Its inductive limit is called
either a global algebra (in the case of von Neumann algebras) or a
quasilocal algebra (for a net of $C^*$-algebras). This
construction is extended to non-Minkowski spaces, e.g., globally
hyperbolic spacetimes \cite{brun,brun2,ruzz}.

In a different formulation of algebraic QFT with reference to the
field-particle dualism, realistic quantum field models are
described by tensor algebras, as a rule.

Let $Q$ be a nuclear space. Let us consider the direct limit
$\wh\otimes Q$ (\ref{x2}) of the vector spaces $\wh\ot^{\leq n} Q$
(\ref{x2a}) where $\wh\ot$ is the topological tensor product with
respect to Grothendieck's topology. As was mentioned above,
provided with the inductive limit topology, the tensor algebra
$\wh\otimes Q$ (\ref{x2}) is a unital nuclear $b^*$-algebra
(Example \ref{gns90}). Therefore, one can apply GNS representation
Theorem \ref{gns36} to it. A state $f$ of this algebra is given by
a tuple $\{f_n\}$ of continuous forms on the tensor algebra $A_Q$
(\ref{gns100}). Its value $f(q^1\cdots q^n)$ are interpreted as
the vacuum expectation of a system of fields $q^1,\ldots,q^n$.

In algebraic QFT, one usually choose by $Q$ the Schwartz space of
functions of rapid decrease.

\begin{remark} \label{spr451} \mar{spr451}
By functions of rapid decrease on an Euclidean space $\mathbb R^n$
are called  complex smooth functions $\psi(x)$ such that the
quantities
\mar{spr453}\beq
\|\psi\|_{k,m}=\op\max_{|\al|\leq k} \op\sup_x(1+x^2)^m|D^\al
\psi(x)| \label{spr453}
\eeq
are finite for all $k,m\in \mathbb N$. Here, we follow the
standard notation
\be
D^\al=\frac{\dr^{|\al|}}{\dr^{\al_1} x^1\cdots\dr^{\al_n}x^n},
\qquad |\al|=\al_1+\cdots +\al_n,
\ee
for an $n$-tuple  of natural numbers $\al=(\al_1,\ldots,\al_n)$.
The functions of rapid decrease constitute a nuclear space
$S(\mathbb R^n)$ with respect to the topology determined by the
seminorms (\ref{spr453}). Its dual is a space $S'(\mathbb R^n)$ of
tempered distributions \cite{bog,gelf64,piet}. The corresponding
contraction form is written as
\be
\lng \psi,h\rng=\op\int \psi(x) h(x) d^nx, \qquad \psi\in
S(\mathbb R^n), \qquad h\in S'(\mathbb R^n).
\ee
A  space $S(\mathbb R^n)$ is provided with a non-degenerate
separately continuous Hermitian form
\be
\lng \psi|\psi'\rng=\int \psi(x)\ol{\psi'(x)}d^nx.
\ee
The completion of $S(\mathbb R^n)$ with respect to this form is a
space $L^2_C(\mathbb R^n)$ of square integrable complex functions
on $\mathbb R^n$. We have a rigged Hilbert space
\be
S(\mathbb R^n)\subset L^2_C(\mathbb R^n) \subset S'(\mathbb R^n).
\ee
Let $\mathbb R_n$ denote the dual of $\mathbb R^n$  coordinated by
$(p_\la)$. The Fourier transform
\mar{spr460,1}\ben
&& \psi^F(p)=\int \psi(x)e^{ipx}d^nx, \qquad px=p_\la x^\la,
\label{spr460}\\
&& \psi(x)=\int \psi^F(p)e^{-ipx}d_np, \qquad d_np=(2\pi)^{-n}d^np,
\label{spr461}
\een
defines an isomorphism between the spaces $S(\mathbb R^n)$ and
$S(\mathbb R_n)$. The Fourier transform of tempered distributions
is given by the condition
\be
\int h(x)\psi(x)d^nx=\int h^F(p)\psi^F(-p)d_np,
\ee
and it is written in the form (\ref{spr460}) -- (\ref{spr461}). It
provides an isomorphism between the spaces of tempered
distributions $S'(\mathbb R^n)$ and $S'(\mathbb R_n)$.
\end{remark}

For the sake of simplicity, we here restrict our consideration to
real scalar fields and choose by $Q$ the real subspace $RS^4$ of
the Schwartz space $S(\mathbb R^4)$ of smooth complex functions of
rapid decrease on $\mathbb R^4$ \cite{ccr}. Since a subset
$\op\ot^nS(\mathbb R^k)$ is dense in $S(\mathbb R^{kn})$, we
henceforth identify the tensor algebra $A_{RS^4}$ (\ref{gns100})
of a nuclear space $RS^4$ with the algebra
\mar{qm801}\beq
A=\mathbb R\oplus RS^4\oplus RS^8\oplus\cdots, \label{qm801}
\eeq
called the Borchers algebra \cite{borch,hor,ccr}. Any state $f$ of
this algebra is represented by a collection of tempered
distributions $\{W_k\in S'(\mathbb R^{4k})\}$ by the formula
\be
f(\psi_k)= \int
W_k(x_1,\ldots,x_k)\psi_k(x_1,\ldots,x_k)d^4x_1\cdots d^4x_k,
\qquad \psi_k\in RS^{4k}.
\ee

For instance, the states of scalar quantum fields in a Minkowski
space $\mathbb R^4$ are described by the Wightman functions
$W_n\subset S'(\mathbb R^{4k})$ in the Minkowski space which obey
the Garding--Wightman axioms of axiomatic QFT \cite{bog,wigh,zin}.
Let us mention the Poincar\'e covariance axiom, the condition of
the existence and uniqueness of a vacuum $\thh_0$, and the
spectrum condition. They imply that: (i) a carrier Hilbert space
$E_W$ of Wightman quantum fields admits a unitary representation
of a Poinar\'e group, (ii) a space $E_W$ contains a unique (up to
scalar multiplications) vector $\psi_0$, called the vacuum vector,
invariant under Poincar\'e transformations, (iii) the spectrum of
an energy-momentum operator lies in the closed positive light
cone. In particular, the Poincar\'e covariance condition implies
the translation invariance and the Lorentz covariance of Wightman
functions. Due to the translation invariance of Wightman functions
$W_k$, there exist tempered distributions $w_k\in S'(\mathbb
R^{4k-4})$, also called Wightman functions, such that
\mar{1278}\beq
W_k(x_1,\ldots,x_k)= w_k(x_1-x_2,\ldots,x_{k-1}-x_k). \label{1278}
\eeq
Note that Lorentz covariant tempered distributions for one
argument only are well described \cite{bog,zin07}. In order to
modify Wightman's theory, one studies different classes of
distributions which Wightman functions belong to
\cite{solov,thom}.

Let us here focus on states of the Borchers algebra $A$
(\ref{qm801}) which describe free quantum scalar fields of mass
$m$ \cite{ccr,axiom}.

Let us provide a nuclear space $RS^4$ with a positive complex
bilinear form
\mar{qm802,x5}\ben
&& (\psi|\psi')_m=\frac{2}{i}\int \psi(x)D^-_m(x-y)\psi'(y)d^4xd^4y=\int
\psi^F(-\om,-\op p^\to)\psi'^F(\om,\op p^\to)\frac{d_3p}{\om},
\label{qm802}\\
&& D^-_m(x)=i(2\pi)^{-3}\int \exp[-ipx]\thh(p_0)\dl(p^2-m^2)d^4p,
\label{x5}\\
&& \om=({\op p^\to}^2 +m^2)^{1/2}, \nonumber
\een
where $p^2$ is the Minkowski square, $\thh(p_0)$ is the Heaviside
function, and $D^-_m(x)$ is the negative frequency part of the
Pauli--Jordan function
\mar{gns120}\beq
 D_m(x)=i(2\pi)^{-3}\int
\exp[-ipx](\thh(p_0)-\thh(-p_0))\dl(p^2-m^2)d^4p. \label{gns120}
\eeq
Since a function $\psi(x)$ is real, its Fourier transform
(\ref{spr460}) satisfies an equality $\psi^F(p)=\ol\psi^F(-p)$.

The bilinear form (\ref{qm802}) is degenerate because the
Pauli--Jordan function $D^-_m(x)$ obeys a mass shell equation
\be
(\Box +m^2)D^-_m(x)=0.
\ee
It takes non-zero values only at elements $\psi^F\in RS_4$ which
are not zero on a mass shell $p^2=m^2$. Therefore, let us consider
the quotient space
\mar{gns121}\beq
\g_m:RS^4\to RS^4/J, \label{gns121}
\eeq
where
\be
J=\{\psi\in RS^4\, :\, (\psi|\psi)_m=0\}
\ee
is the kernel of the square form (\ref{qm802}). The map $\g_m$
(\ref{gns121}) assigns  the couple of functions $(\psi^F(\om,\op
p^\to),\psi^F(-\om,\op p^\to))$ to each element $\psi\in RS^4$
with a Fourier transform $\psi^F(p_0,\op p^\to)\in RS_4$. Let us
equip the factor space $RS^4/J$ with a real bilinear form
\mar{qm803}\ben
&& (\g\psi|\g\psi')_L=\mathrm{Re}(\psi|\psi')= \label{qm803}\\
&& \qquad \frac12\int [\psi^F(-\om,-\op p^\to)\psi'^F(\om,\op p^\to)
+\psi^F(\om,-\op p^\to) \psi'^F(-\om,\op p^\to)]\frac{d_3\op
p^\to}{\om}. \nonumber
\een
Then it is decomposed into a direct sum $RS^4/J=L^+\oplus L^-$ of
subspaces
\be
L^\pm=\{\psi^F_\pm(\om,\op p^\to)=\frac12(\psi^F(\om,\op p^\to)
\pm\psi^F(-\om,\op p^\to))\},
\ee
which are mutually orthogonal with respect to the bilinear form
(\ref{qm803}).

There exist continuous isometric morphisms
\be
\g_+:\psi^F_+(\om,\op p^\to) \to q^F(\op p^\to)=\om^{-1/2}\psi^F_+
(\om,\op p^\to),\qquad \g_-:\psi^F_-(\om,\op p^\to) \to q^F(\op
p^\to)=-i\om^{-1/2}\psi^F_- (\om,\op p^\to)
\ee
of spaces $L^+$ and $L^-$ to a nuclear space $RS^3$ endowed with a
non-degenerate separately continuous Hermitian form
\mar{qm807}\beq
\lng q|q'\rng=\int q^F(-\op p^\to)q'^F(\op p^\to)d_3p.
\label{qm807}
\eeq
It should be emphasized that the images $\g_+(L^+)$ and
$\g_-(L^-)$ in $RS^3$ are not orthogonal with respect to the
scalar form (\ref{qm807}). Combining $\g_m$ (\ref{gns121}) and
$\g_\pm$, we obtain continuous morphisms $\tau_\pm: RS^4\to RS^3$
given by the expressions
\be
&& \tau_+(\psi)=\g_+(\g_m\psi)_+=\frac{1}{2\om^{1/2}}\int[\psi^F(\om,\op
p^\to) + \psi^F(-\om,\op p^\to)]\exp[-i\op p^\to\op x^\to]d_3p,\\
&& \tau_-(\psi)=\g_-(\g_m\psi)_-=\frac{1}{2i\om^{1/2}}\int[\psi^F(\om,\op
p^\to) - \psi^F(-\om,\op p^\to)]\exp[-i\op p^\to\op x^\to]d_3p.
\ee

Now let us consider a Heisenberg CCR algebra
\mar{1110}\beq
\ccG(RS^3)=\{(\f(q),\pi(q)), I,\,q\in RS^3\} \label{1110}
\eeq
modelled over a nuclear space $RS^3$, which is equipped with the
Hermitian form (\ref{qm807}) (Section 8). Using the morphisms
$\tau_\pm$, let us define a map
\mar{1111}\beq
\G_m: RS^4\ni \psi \to \f(\tau_+(\psi)) -\pi(\tau_-(\psi))\in
\ccG(RS^3). \label{1111}
\eeq
With this map, one can think of (\ref{1110}) as being the algebra
of instantaneous CCR of scalar fields on a Minkowski space
$\mathbb R^4$. Owing to the map (\ref{1111}), any representation
of the Heisenberg CCR algebra $\ccG(RS^3)$ (\ref{1110})  defined
by a translationally quasi-invariant measure $\m$ on $S'(\mathbb
R^3)$ induces a state
\mar{qm805}\beq
f_m(\psi^1\cdots\psi^n)=\lng\f(\tau_+(\psi^1))
+\pi(\tau_-(\psi^1))] \cdots [\f(\tau_+(\psi^n))
+\pi(\tau_-(\psi^n))]\rng \label{qm805}
\eeq
of the Borchers algebra $A$ (\ref{qm801}). Furthermore, one can
justify that the corresponding distributions $W_n$ fulfil the mass
shell equation and that the following commutation relation holds:
\be
W_2(x,y) -W_2(y,x)=-iD_m(x-y),
\ee
where $D_m(x-y)$ is the Pauli--Jordan function (\ref{gns120}).
Thus, the state (\ref{qm805}) of the Borchers algebra $A$
(\ref{qm801}) describes quantum scalar fields of mass $m$.

For instance, let us consider the Fock representation $Z_{\mathrm
F}(q)$ (\ref{qm567}) of the Heisenberg CCR algebra $\ccG(RS^3)$
(\ref{1110}). Using the formulae in Remark \ref{ccr10} where a
form $\lng q|q'\rng$ is given by the expression (\ref{qm807}), one
observes that the state $f_m$ (\ref{qm805}) satisfies the Wick
theorem relations
\mar{1263}\beq
f_m(\psi^1\cdots\psi^n)=\op\sum_{(i_1\ldots i_n)}
f_2(\psi^{i_1}\psi^{i_2})\cdots f_2(\psi^{i_{n-1}}\psi^{i_n}),
\label{1263}
\eeq
where a state $f_2$ is given by the Wightman function
\mar{ccr21}\beq
W_2(x,y)=\frac{1}{i}D^-_m(x-y). \label{ccr21}
\eeq
Thus, the state $f_m$ (\ref{1263}) describes free quantum scalar
fields of mass $m$.

Similarly, one can obtain states of the Borchers algebra $A$
(\ref{qm801}) generated by non-Fock representations (\ref{qm563})
of the instantaneous CCR algebra $\ccG(RS^3)$, e.g., if
$K^{-1}=c\bb\neq 2^{-1/2}\bb$. These states fail to be defined by
Wightman functions.

It should be emphasized that, given a different mass $m'$, we have
a different map $\G_{m'}$ (\ref{1111}) of the Borchers algebra $A$
(\ref{qm801}) to the Heisenberg CCR algebra $\ccG(RS^3)$
(\ref{1110}). Accordingly, the Fock representation $Z_{\mathrm
F}(q)$ (\ref{qm567}) of the Heisenberg CCR algebra $\ccG(RS^3)$
(\ref{1110}) yields the state $f_{m'}$ (\ref{1263}) where a state
$f_2$ is given by the Wightman function
\mar{ccr21'}\beq
W_2(x,y)=\frac{1}{i}D^-_{m'}(x-y). \label{ccr21'}
\eeq
If $m\neq m'$, the states $f_m$ and $f_{m'}$ (\ref{1263}) the
Borchers algebra $A$ (\ref{qm801}) are inequivalent because its
representations $\G_m$ and $\G_{m'}$ (\ref{1111}) possess
different kernels.

\section{Euclidean QFT}

In QFT, interacting quantum fields created at some instant and
annihilated at another one are described by complete Green
functions.  They are given by the chronological functionals
\mar{1260,030}\ben
&&f^c(\psi_k)= \int
W_k^c(x_1,\ldots,x_k)\psi_k(x_1,\ldots,x_k)d^4x_1\cdots d^4x_k,
\qquad \psi_k\in RS^{4k}, \label{1260}\\
&& W^c_k(x_1,\ldots,x_k)= \op\sum_{(i_1\ldots
i_k)}\thh(x^0_{i_1}-x^0_{i_2})
\cdots\thh(x^0_{i_{k-1}}-x^0_{i_n})W_k(x_1,\ldots,x_k),
\label{030}
\een
where $W_k\in S'(\mathbb R^{4k})$ are tempered distributions, and
the sum runs through all permutations $(i_1\ldots i_k)$ of the
tuple of numbers $1,\ldots,k$ \cite{bog2}.

A problem is that the functionals $W^c_k$ (\ref{030}) need not be
tempered distributions. For instance, $W^c_1\in S'(\mathbb R)$ iff
$W_1\in S'(\mathbb R_\infty)$, where $\mathbb R_\infty$ is the
compactification of $\mathbb R$ by means of a point
$\{+\infty\}=\{-\infty\}$ \cite{bog}. Moreover, chronological
forms are not positive. Therefore, they do not provide states of
the Borchers algebra $A$ (\ref{qm801}) in general.

At the same time, the chronological forms (\ref{030}) come from
the Wick rotation of Euclidean states of the Borchers algebra
\cite{sard91,ccr,axiom}. As is well known, the Wick rotation
enables one to compute the Feynman diagrams of perturbed QFT by
means of Euclidean propagators. Let us suppose that it is not a
technical trick, but quantum fields in an interaction zone are
really Euclidean. It should be emphasized that the above mentioned
Euclidean states differ from the well-known Schwinger functions in
the Osterwalder--Shraded Euclidean QFT
\cite{bog,ost,ccr,axiom,schlin,zin}. The Schwinger functions are
the Laplace transform of Wightman functions, but not chronological
forms.

Since the chronological forms (\ref{030}) are symmetric, the
Euclidean states of a Borchers algebra $A$ can be obtained as
states of the corresponding commutative tensor algebra $B_{RS^4}$
(\ref{gns101}) \cite{sard91,ccr,axiom}. Provided with the direct
sum topology, $B_{RS^4}$ becomes a topological involutive algebra.
It coincides with the enveloping algebra of the Lie algebra of an
additive Lie group $T(RS^4)$ of translations in $RS^4$. Therefore,
one can obtain states of an algebra $B_{RS^4}$ by constructing
cyclic strongly continuous unitary representations of a nuclear
Abelian group $T(RS^4)$ (Section 7). Such a representation is
characterized by a continuous positive-definite generating
function $Z$ on $SR^4$. By virtue of Bochner Theorem \ref{spr525},
this function is the Fourier transform
\mar{031}\beq
Z(\phi)=\int \exp[i \langle\phi,y\rangle]d\mu(y) \label{031}
\eeq
of a positive measure $\mu$ of total mass 1 on the dual $(RS^4)'$
of $RS^4$. Then the above mentioned representation of $T(RS^4)$
can be given by operators
\mar{qq2}\beq
\wh\phi \rho(y)=\exp[i\langle \f,y\rangle]\rho(y)  \label{qq2}
\eeq
in a Hilbert space $L_{\mathbb C}^2((RS^4)',\m)$ of the
equivalence classes of square $\m$-integrable complex functions
$\rho(y)$ on $(RS^4)'$. A cyclic vector $\thh$ of this
representation is the $\m$-equivalence class $\thh\ap_\m 1$ of the
constant function $\rho(y)=1$.

Conversely, every positive measure $\m$ of total mass 1 on the
dual $(RS^4)'$ of $RS^4$ defines the cyclic strongly continuous
unitary representation (\ref{qq2}) of a group $T(RS^4)$. Herewith,
distinct generating functions $Z$ and $Z'$ characterize equivalent
representations $T_Z$ and $T_{Z'}$ (\ref{qq2}) of $T(RS^4)$ in the
Hilbert spaces $L^2_{\mathbb C}((RS^4)',\m)$ and $L^2_{\mathbb
C}((RS^4)',\m')$  iff they are the Fourier transform of equivalent
measures on $(RS^4)'$ (Theorem \ref{gns77}).

If a generating function $Z$ obeys the analiticity condition in
Remark \ref{gns109},  a state $f$ of $B_{RS^4}$ is given by the
expression
\mar{w0}\beq
f_k(\phi_1\cdots\phi_k)=i^{-k}\frac{\dr}{\dr \al^1}
\cdots\frac{\dr}{\dr\alpha^k}Z(\alpha^i\phi_i)|_{\alpha^i=0}=\int\langle
\phi_1,y\rangle\cdots\langle \phi_k,y \rangle d\mu(y). \label{w0}
\eeq
Then one can think of $Z$ (\ref{031}) as being a generating
functional of complete Euclidean Green functions $f_k$ (\ref{w0}).

For instance, free Euclidean fields are described by Gaussian
states. Their generating functions are of the form
\mar{a10}\beq
Z(\phi)=\exp(-\frac12M(\f,\f)), \label{a10}
\eeq
where $M(\f,\f)$ is a positive-definite Hermitian bilinear form on
$RS^4$ continuous in each variable. They are the Fourier transform
of some Gaussian measure on $(RS^4)'$. In this case, the forms
$f_k$ (\ref{w0}) obey the Wick relations (\ref{ccr20}) where
\be
f_1=0, \qquad f_2(\f,\f')=M(\f,\f').
\ee
Furthermore, a covariance form $M$ on $RS^4$ is uniquely
determined as
\mar{1270}\beq
M(\f_1,\f_2)=\int W_2(x_1,x_2)\f_1(x_1)\f_2(x_2) \label{1270}
d^nx_1d^nx_2.
\eeq
by a tempered distribution $W_2\in S'(\mathbb R^8)$.

In particular, let a tempered distribution $M(\f,\f')$ in the
expression (\ref{1270}) be Green's function of some positive
elliptic differential operator $\cE$, i.e.,
\be
\cE_{x_1}W_2(x_1,x_2)=\dl(x_1-x_2),
\ee
where $\dl$ is Dirac's $\dl$-function. Then the distribution $W_2$
reads
\mar{1272}\beq
W_2(x_1,x_2)=w(x_1-x_2), \label{1272}
\eeq
and we obtain a form
\be
&& f_2(\f_1\f_2)=M(\f_1,\f_2)=\int w(x_1-x_2)\f_1(x_1)\f_2(x_2)
d^4x_1 d^4x_2=\\
&& \qquad \int w(x)\f_1(x_1)\f_2(x_1-x)d^4x d^4x_1=\int w(x)\vf(x)d^4x=
\int w^F(p)\vf^F(-p) d_4p, \\
&& x=x_1-x_2, \qquad \vf(x)=\int \f_1(x_1)\f_2(x_1-x)d^4x_1.
\ee
For instance, if
\be
\cE_{x_1} =-\Delta_{x_1}+m^2,
\ee
where $\Delta$ is the Laplacian, then
\mar{1271}\beq
w(x_1-x_2)=\int\frac{\exp(-iq(x_1-x_2))}{p^2+m^2}d_4p,
\label{1271}
\eeq
where $p^2$ is the Euclidean square, is the propagator of a
massive Euclidean scalar field. Note that, restricted to the
domain $(x^0_1-x^0_2)<0$, it coincides with the Schwinger function
$s_2(x_1-x_2)$.

Let $w^F$ be the Fourier transform  of the distribution $w$
(\ref{1272}). Then its Wick rotation is the functional
\be
\wt w(x)=\thh(x)\op\int_{\ol Q_+}w^F(p)\exp(-px)d_4p +
\thh(-x)\op\int_{\ol Q_-}w^F(p)\exp(-px)d_4p
\ee
on scalar fields on a Minkowski space \cite{ccr,axiom}. For
instance, let $w(x)$ be the Euclidean propagator (\ref{1271}) of a
massive scalar field. Then due to the analyticity of
\be
w^F(p)=(p^2+m^2)^{-1}
\ee
on the domain $\im p\cdot \re p>0$, one can show that $\wt
w(x)=-iD^c(x)$ where $D^c(x)$ is familiar causal Green's function.

A problem is that a measure $\m$ in the generating function $Z$
(\ref{031}) fails to be written in an explicit form.

At the same time, a measure $\mu$ on $(RS^4)'$ is uniquely defined
by a set of measures $\mu_N$ on the finite-dimensional spaces
$\mathbb R_N=(RS^4)'/E$ where $E\subset (RS^4)'$ denotes a
subspace of forms on $RS^4$ which are equal to zero at some
finite-dimensional subspace $\mathbb R^N\subset RS^4$. The
measures $\mu_N$ are images of $\mu$ under the canonical mapping
$(RS^4)'\to\mathbb R_N$. For instance, every vacuum expectation
$f_n(\f_1\cdots\f_n)$ (\ref{w0}) admits the representation by an
integral
\mar{S8}\beq
f_k(\f_1\cdots\f_k)=\int\langle w,\f_1\rangle\cdots\langle w,\f_k
\rangle d\mu_N(w) \label{S8}
\eeq
for any finite-dimensional subspace $\mathbb R^N$ which contains
$\f_1,\ldots,\f_k$. In particular, one can replace the generating
function (\ref{031}) by the generating function
\be
Z_N(\la_ie^i)=\int\exp(i\la_iw^i)\mu_N(w^i)
\ee
on $\mathbb R^N$ where $\{e^i\}$ is a basis for $\mathbb R^N$ and
$\{w^i\}$ are coordinates with respect to the dual basis for
$\mathbb R_N$. If $f$ is a Gaussian state, we have the familiar
expression (\ref{qm610}):
\mar{S9}\beq
d\mu_N=(2\pi\det[M^{ij}])^{-N/2}\exp[-\frac12(M^{-1})_{ij}w^iw^j]
d^Nw, \label{S9}
\eeq
where $M^{ij}=M(e^i,e^j)$ is a non-degenerate covariance matrix.

The representation (\ref{S8}) however is not unique, and the
measure $\mu_N$ depends on the specification of a
finite-dimensional subspace $\mathbb R^N$ of $RS^4$.

\begin{remark} \label{gns133} \mar{gns133}
Note that an expression
\mar{gns134}\beq
\exp(-\int L(\f)d^4x)\op\prod_x[d\f(x)] \label{gns134}
\eeq
conventionally used in perturbed QFT is a symbolic functional
integral, but not a true measure \cite{glimm,john,schmit}. In
particular, it is translationally invariant, i.e.,
\be
[d\f(x)]=[d(\f(x)+ \mathrm{const}.)],
\ee
whereas there is no (translationally invariant) Lebesgue measure
on infinite-dimensional vector space as a rule (see \cite{versh}
for an example of such a measure).
\end{remark}

\section{Higgs vacuum}

In contrast to the formal expression (\ref{gns134}) of perturbed
QFT, the true integral representation (\ref{031}) of generating
functionals enables us to handle non-Gaussian and inequivalent
Gaussian representations of the commutative tensor algebra $B_Q$
(\ref{gns101}) of Euclidean scalar fields. Here, we describe one
of such a representation as a model of  a Higgs vacuum
\cite{NC91,sard91}.

In Standard Model of particle physics, Higgs vacuum is represented
as a constant background part $\si_0$ of a Higgs scalar field
$\si$ \cite{nov,SM}. In algebraic QFT, one can describe free Higgs
field similar to matter fields by a commutative tensor algebra
$B_\Si$ where $\Si$ is a real nuclear space.

Let $Z(\wh\si)$ be the generating function (\ref{a10}) of a
Gaussian state of $B_\Si$, and let $\m$ be the corresponding
Gaussian measure on the dual $\Si'$ of $\Si$. In contrast with a
finite-dimensional case, Gaussian measures on infinite-dimensional
spaces  fail to be quasi-invariant under translations as a rule.
The introduction of a Higgs vacuum means a translation
\be
\g:\Si'\ni \si\to \si+\si_0\in\Si', \qquad \si_0\in\Si'
\ee
in a space $\Si'$ such that an original Gaussian measure $\m(\si)$
is replaced by a measure $\mu_{\si_0}(\si)=\m(\si+\si_0)$
possessing the Fourier transform
\be
Z_{\si_0}(\wh\si)=\exp(i\lng\wh\si,\si_0\rng Z(\wh\si)
\ee
The measures $\mu$ and $\mu_{\si_0}$ are equivalent iff a vector
$\si_0\in\Si'$ belongs to the canonical image of $\Si$ in $\Si'$
with respect to the scalar form $\langle|\rangle=M(,)$ (Remark
\ref{gns79}). Then the measures $\mu$ and $\mu_{\si_0}$ define the
equivalent states (\ref{w0}) of an algebra $B_\Si$. This
equivalence is performed by the unitary operator
\be
\rho(\si)\to \exp(-\langle \si|\si_0\rangle)\rho(\si+\si_0),
\qquad \rho(\si)\in L^2(\Si',\mu).
\ee
This operator fails to be constructed if $\si_0\in
\Si'\setminus\Si$, and the measures $\m$ and $\m_{\si_0}$ are
inequivalent.

Following the terminology of Standard Model, let us call $\si_0\in
\Si'\setminus\Si$ the Higgs vacuum field and $\si\in\Si\subset
\Si'$ the Higgs boson fields. Then we can say the following.

(i) A Higgs vacuum field $\si_0$ and Higgs boson fields $\si$
belong to different classes of functions. For instance, one
usually chooses a constant Higgs vacuum field $\si_0$ in QFT. If
$\Si=RS^4$, a constant function is an element of $(SR^4)'\setminus
SR^4$. At the same time, since $\Si$ is dense in $\Si'$, the
elements $\si_0$ and $\si$ can be arbitrarily closed to each other
with respect to a topology in $\Si'$. However, a covariance form
$M$ and some other functions being well defined at points $\si$
become singular at points $\si_0$.

(ii) One can think of a Higgs vacuum field $\si_0$ as being the
classical one in the sense that $\si_0\in \Si'\setminus\Si$,
whereas Higgs boson fields $\si\in \Si\subset \Si'$ are quantized
fields because the possess quantum partners $\wh\si\in\Si\subset
B_\Si$.

(iii) States of Higgs boson fields in the presence of in
equivalent Higgs vacua are inequivalent.

(iv) Let the generating function $Z$ and $Z_{\si_0}$ be restricted
to some finite-dimensional subspace $\mathbb R^N\subset\Si$. Then
there exists an element $\si_{0N}\in\mathbb R_N$ such that
$\lng\wh\si,\si_0\rng=\lng\wh\si|\si_{0N}\rng$ for any
$\wh\si\in\mathbb R^N$. As a consequence, a generating function
$Z_{\si_0}$ takes a form
\be
Z_{\si_0N}(\la_i\wh\si^i)=
(2\pi\det[M^{ij}])^{-N/2}\int\exp(\la_i\si^i)
\exp[-\frac12(M^{-1})_{ij}(\si-\si_{0N})^i(\si -\si_{0N})^j]
d^N\si,
\ee
where $M^{ij}$ is the covariance matrix of $Z_N$, $\si^i$ denote
coordinates in $\mathbb R_N$, and $\si^i_{0N}$ are coordinates of
a vector $\si_{0N}$ in $\mathbb R_N$. It follows that, if a number
of quantum Higgs boson fields $\\wh\si$ is finite, their
interaction with a classical Higgs vacuum field $\si_0$ reduced to
an interaction with some quantum fields $\wh\si_{0N}$ by
perturbation theory.

In Standard Model, a Higgs vacuum is responsible for spontaneous
symmetry breaking. Let us study this phenomenon (Sections 12 and
13).

\section{Automorphisms of quantum systems}

In order to say something, we mainly restrict our analysis to
automorphisms of $C^*$ algebras.

We consider uniformly and strongly continuous one-parameter groups
of automorphisms of $C^*$-algebras. Let us note that any weakly
continuous one-parameter group of endomorphism of a $C^*$-algebra
also is also strongly continuous and their weak and strong
generators coincide with each other \cite{brat,book05}.

\begin{remark} \label{w416} \mar{w416}
There is the following relation between morphisms of a
$C^*$-algebra $A$ and a set $F(A)$ of its states which is a convex
subset of the dual $A'$ of $A$ (Theorem \ref{gns17}).  A linear
morphism $\g$ of a $C^*$-algebra $A$ as a vector  space is called
the Jordan morphism if relations
\be
\g(ab+ba)=\g(a)\g(b)+\g(b)\g(a), \qquad \f(a^*)=\g(a)^*, \qquad
a,b\in A.
\ee
hold. One can show the following \cite{emch}. Let $\g$ be a Jordan
automorphism of a unital $C^*$-algebra $A$. It yields the dual
weakly$^*$ continuous affine bijection $\g'$ of $F(A)$ onto
itself, i.e.,
\be
\g'(\la f+(1-\la)f')=\la\g'(f) +(1-\la)\g'(f'), \qquad f,f',\in
F(A), \qquad \la\in [0,1].
\ee
Conversely, any such a map of $F(A)$ is the dual to some Jordan
automorphism of $A$. However, if $G$ is a connected group of
weakly continuous Jordan automorphisms of a unital $C^*$-algebra
$A$ is a weakly (and, consequently, strongly) continuous group of
automorphisms of $A$.
\end{remark}

A topological group $G$ is called the strongly (resp. uniformly)
continuous group of automorphisms of a $C^*$-algebra $A$ if there
is its continuous monomorphism to the group $\Is(A)$ of
automorphisms of $A$ provided with the strong (resp. normed)
operator topology, and if its action on $A$ is separately
continuous.

One usually deals with strongly continuous groups of automorphisms
because of the following reason. Let $G(\mathbb R)$ be
one-parameter group of automorphisms of a $C^*$-algebra $A$. This
group  is uniformly (resp. strongly) continuous if it is a range
of a continuous map of $\mathbb R$ to the group $\Is(A)$ of
automorphisms of $A$ which is provided with the normed (resp.
strong) operator topology and whose action on $A$ is separately
continuous. A problem is that, if a curve $G(\mathbb R)$ in
$\Is(A)$ is continuous with respect to the normed operator
topology, then a curve $G(\mathbb R)(a)$ for any $a\in A$ is
continuous in a $C^*$-algebra $A$, but the converse is not true.
At the same time, a curve $G(\mathbb R)$ is continuous in $\Is(A)$
with respect to the strong operator topology iff a curve
$G(\mathbb R)(a)$ for any $a\in A$  is continuous in $A$. By this
reason, strongly continuous one-parameter groups of automorphisms
of $C^*$-algebras are most interesting. However, the infinitesimal
generator of such a group fails to be bounded, unless this group
is uniformly continuous.

\begin{remark} \label{w431} \mar{w431}
If $G(\mathbb R)$ is a strongly continuous one-parameter group of
automorphisms of a $C^*$-algebra $A$, there are the following
continuous maps \cite{brat}:

$\bullet$ $\mathbb R\ni t\to \lng G_t(a), f\rng \in \mathbb C$ is
continuous for all $a\in A$ and $f\in A'$;

$\bullet$ $A\ni a\to G_t(a) \in A$ is continuous for all
$t\in\mathbb R$;

$\bullet$ $\mathbb R\ni t\to G_t(a)\in A$ is continuous for all
$a\in A$.
\end{remark}

Without a loss of generality, we further assume that $A$ is a
unital $C^*$-algebra. Infinitesimal generators of one-parameter
groups of automorphisms of $A$ are derivations of $A$.

By a derivation $\dl$ of $A$ throughout is meant an (unbounded)
symmetric derivation of $A$ (i.e., $\dl(a^*)=\dl(a)^*$, $a\in A$)
which is defined on a dense involutive subalgebra $D(\dl)$ of $A$.
If a derivation $\dl$ on $D(\dl)$ is bounded, it is extended to a
bounded derivation everywhere on $A$. Conversely, every derivation
defined everywhere on a $C^*$-algebra is bounded \cite{dixm}. For
instance, any inner derivation $\dl(a)=i[b,a]$, where $b$ is a
Hermitian element of $A$, is bounded. A space of derivations of
$A$ is provided with the involution $u\to u^*$ defined by the
equality
\mar{w160}\beq
\dl^*(a)= -\dl(a^*)^*, \qquad a\in\cA. \label{w160}
\eeq

There is the following relation between bounded derivations of a
$C^*$-algebra $A$ and uniformly continuous one-parameter groups of
automorphisms of $A$ \cite{brat}.

\begin{theorem} \label{spr591} \mar{spr591}
Let $\dl$ be a derivation of a $C^*$-algebra $A$. The following assertions
are equivalent:

$\bullet$ $\dl$ is defined everywhere and, consequently, is bounded;

$\bullet$ $\dl$ is the infinitesimal generator of a uniformly
continuous one-parameter group $G(\mathbb R)$ of automorphisms of
a $C^*$-algebra $A$.

\noindent Furthermore, for any representation $\pi$ of $A$ in a
Hilbert space $E$, there exists a bounded self-adjoint operator
$\cH\in \pi(A)''$ in $E$ and the unitary uniformly continuous
representation
\mar{spr579'}\beq
\pi(G_t)=\exp(-it\cH), \qquad t\in\mathbb R, \label{spr579'}
\eeq
of the group $G(\mathbb R)$ in $E$ such that
\mar{spr592,3}\ben
&& \pi(\dl(a))=-i[\cH,\pi(a)], \qquad  a\in A, \label{spr592}\\
&& \pi(G_t(a))=e^{-it\cH}\pi(a)e^{it\cH}, \qquad  t\in\mathbb
R.\label{spr593}
\een
\end{theorem}

A $C^*$-algebra need not admit non-zero bounded derivations. For
instance, no commutative $C^*$-algebra possesses bounded
derivations.
 The
following is the relation between (unbounded) derivations of a
$C^*$-algebra $A$ and strongly continuous
one-parameter groups of automorphisms of $A$ \cite{brat75,pow}.

\begin{theorem} \label{spr595} \mar{spr595}
Let $\dl$ be a closable derivation of a $C^*$-algebra $A$. Its
closure $\ol\dl$ is an infinitesimal generator of a strongly
continuous one-parameter group of automorphisms of $A$ iff

(i) the set $(\bb +\la\dl)(D(\dl)$ for any $\la\in\mathbb
R\setminus\{0\}$ is dense in $A$,

(ii)
$\|(\bb +\la\dl)(a)\|\geq \|a\|$ for any $\la\in\mathbb R$ and any $a\in
A$.
\end{theorem}

It should be noted that, if
$A$ is a unital algebra
and $\dl$ is its closable derivation, then $\bb\in D(\dl)$.

Let us mention a more convenient sufficient condition of a
derivation of a $C^*$-algebra to be an infinitesimal generator of
a strongly continuous one-parameter group of its automorphisms. A
derivation $\dl$ of a $C^*$-algebra $A$ is called well-behaved if,
for each element $a\in D(\dl)$, there exists a state $f$ of $A$
such that $f(a)=\|a\|$ and $f(\dl(a))=0$. If $\dl$ is a
well-behaved derivation, it is closable \cite{kish}, and obeys the
condition (ii) in Theorem \ref{spr595} \cite{brat75,pow}. Then we
come to the following.

\begin{theorem} \label{spr596} \mar{spr596}
If $\dl$ is a well-behaved derivation of a $C^*$-algebra $A$ and
it obeys the condition (i) in Theorem \ref{spr595}, its closure
$\ol\dl$ is an infinitesimal generator of a strongly continuous
one-parameter group of automorphisms of $A$.
\end{theorem}

For instance, a derivation $\dl$ is well-behaved if it is
approximately inner, i.e., there exists a sequence of self-adjoint
elements $\{b_n\}$ in $A$ such that
\be
\dl(a)=\op\lim_n i[b_n,a], \qquad  a\in A.
\ee

In contrast with a case of a uniformly continuous one-parameter
group of automorphisms of a $C^*$-algebra $A$, a representation of
$A$ does not imply necessarily a unitary representation
(\ref{spr579'}) of a strongly continuous one-parameter group of
automorphisms of $A$, unless the following.

\begin{theorem} \label{spr759} \mar{spr759}
Let $G(\mathbb R)$ be a strongly continuous one-parameter group of
automorphisms of a $C^*$-algebra $A$ and $\dl$ its infinitesimal
generator. Let $A$ admit a state $f$ such that
\mar{w420}\beq
|f(\dl(a))|\leq\la[f(a^*a) + f(aa^*)]^{1/2} \label{w420}
\eeq
for all $a\in A$ and a positive number $\la$, and let
$(\pi_f,\thh_f)$ be a cyclic representation of $A$ in a Hilbert
space $E_f$  defined by $f$. Then there exist a self-adjoint
operator $\cH$ on a domain $D(\cH)\subset A\thh_f$ in $E_f$ and
the unitary strongly continuous representation (\ref{spr579'}) of
$G(\mathbb R)$ in $E_f$ which fulfils the relations (\ref{spr592})
-- (\ref{spr593}) for $\pi=\pi_f$.
\end{theorem}

It should be emphasized that the condition (\ref{w420}) in Theorem
\ref{spr759} is sufficient in order that a derivation $\dl$ to be
closable \cite{kish}.

Note that there is a general problem of a unitary representation
of an automorphism group of a $C^*$-algebra $A$. Let $\pi$ be a
representation of $A$ in a Hilbert space $E$. Then an automorphism
$\rho$ of $A$ possesses a unitary representation in $E$ if there
exists a unitary operator $U_\rho$ in $E$ such that
\mar{081}\beq
\pi(\rho(a))=U_\rho\pi(a)U_\rho^{-1}, \qquad a\in A. \label{081}
\eeq
A key point is that such a representation is never unique. Namely,
let $U$ and $U'$ be arbitrary unitary elements of the commutant
$\pi(A)'$ of $\pi(A)$. Then $UU_\rho U'$ also provides a unitary
representation of $\rho$. For instance, one can always choose
phase multipliers $U=\exp(i\al)\bb\in U(1)$. A consequence of this
ambiguity is the following.

Let $G$ be a group of automorphisms of an algebra $A$ whose
elements $g\in G$ admit unitary representations $U_g$ (\ref{081}).
The set of operators $U_g$, $g\in G$, however need not be a group.
In general, we have
\be
U_g U_{g'}=U(g,g')U_{gg'}U'(g,g'), \qquad U(g,g'),U'(g,g')\in
\pi(A)'.
\ee
If all $U(g,g')$ are phase multipliers, one says that the unitary
operators $U_g$, $g\in G$, form a projective representation
$U(G)$:
\be
 U_g U_{g'}=k(g,g')U_{gg'}, \qquad g,g'\in G,
\ee
of a group $G$ \cite{cass,varad}. In this case, a set $U(1)\times
U(G)$ becomes a group which is a central $U(1)$-extension
\mar{084}\beq
\bb\ar U(1)\ar U(1)\times U(G)\ar G\ar \bb \label{084}
\eeq
of a group $G$. Accordingly, the projective representation
$\pi(G)$ of $G$ is a splitting of the exact sequence (\ref{084}).
It is characterized by $U(1)$-multipliers $k(g,g')$ which form a
two-cocycle
\mar{085}\beq
k(\bb,g)=k(g,\bb)=\bb, \qquad
k(g_1,g_2g_3)k(g_2,g_3)=k(g_1,g_2)k(g_1g_2,g_3) \label{085}
\eeq
of the cochain complex of $G$ with coefficients in $U(1)$
\cite{book05,mcl}. A different splitting of the exact sequence
(\ref{084}) yields a different projective representation $U'(G)$
of $G$ whose multipliers $k'(g,g')$ form a cocycle equivalent to
the cocycle (\ref{085}). If this cocycle is a coboundary, there
exists a splitting of the extension (\ref{084}) which provides a
unitary representation of a group $G$ of automorphisms of an
algebra $A$ in $E$.

For instance, let $B(E)$ be a $C^*$-algebra of bounded operators
in a Hilbert space $E$. All its automorphisms are inner. Any
(unitary) automorphism $U$ of a Hilbert space $E$ yields an inner
automorphism
\mar{spr750}\beq
a\to UaU^{-1}, \qquad  a\in B(E), \label{spr750}
\eeq
of $B(E)$. Herewith, the automorphism (\ref{spr750}) is the
identity iff $U=\la\bb$, $|\la|=1$, is a scalar operator in $E$.
It follows that the group of automorphisms of $B(E)$ is the
quotient $U(E)/U(1)$ of a unitary group $U(E)$ with respect to a
circle subgroup $U(1)$. Therefore, given a group $G$ of
automorphisms of the $C^*$-algebra $B(E)$, the representatives
$U_g$ in $U(E)$ of elements $g\in G$ constitute a group up to
phase multipliers, i.e.,
\be
U_gU_{g'}=\exp[i\al(g,g')]U_{gg'}, \qquad \al(g,g')\in\mathbb R.
\ee
Nevertheless, if $G$ is a one-parameter weakly$^*$ continuous
group of automorphisms of $B(E)$ whose infinitesimal generator is
a bounded derivation of $B(E)$, one can choose the multipliers
$\exp[i\al(g,g')]=1$.

In a general setting, let $G$ be a group and $\cA$ a commutative
algebra. An $\cA$-multiplier of $G$ is a map $\xi: G\times G\to
\cA$ such that
\be
\xi(\bb_G,g)=\xi(g,\bb_G)=\bb_\cA, \qquad
\xi(g_1,g_2g_3)\xi(g_2,g_3)=\xi(g_1,g_2)\xi(g_1g_2,g_3), \qquad
g,g_i\in G.
\ee
For instance, $\xi:G\times G\to\bb_\cA\in \cA$ is a multiplier.
Two $A$-multipliers $\xi$ and $\xi'$ are said to be equivalent if
there exists a map $f:G\to \cA$ such that
\be
\xi(g_1,g_2)=\frac{f(g_1g_2)}{f(g_1)f(g_2)}\xi'(g_1,g_2), \qquad g_i\in
G.
\ee
An $\cA$-multiplier is called exact if it is equivalent to the
multiplier $\xi=\bb_\cA$. A set of $\cA$-multipliers is an Abelian
group with respect to the pointwise multiplication, and the set of
exact multipliers is its subgroup. Let $HM(G,\cA)$ be the
corresponding factor group.

If $G$ is a locally compact topological group and $\cA$ a Hausdorff
topological algebra,
one additionally requires that multipliers $\xi$ and equivalence maps
$f$ are measurable maps. In this case, there is a natural topology on
$HM(G,\cA)$ which is locally quasi-compact, but need not be Hausdorff
\cite{moor}.

\begin{theorem} \label{w711} \mar{w711} \cite{cass}.
Let $G$ be a simply connected locally compact Lie group. Each
$U(1)$-multiplier $\xi$ of $G$ is brought into a form
$\xi=\exp{i\al}$, where $\al$ is an $\mathbb R$-multiplier.
Moreover, $\xi$ is exact iff $\al$ is well. Any $\mathbb
R$-multiplier of $G$ is equivalent to a smooth one.
\end{theorem}

Let $G$ be a locally compact group of strongly continuous
automorphisms of a $C^*$-algebra $A$. Let $M(A)$ denote a
multiplier algebra of $A$, i.e., the largest $C^*$-algebra
containing $A$ as an essential ideal,  i.e., if $a\in M(A)$ and
$ab=0$ for all $b\in A$, then $a=0$ \cite{woron94}). For instance,
$M(A)=A$ if $A$ is a unital algebra. Let $\xi$ be a multiplier of
$G$ with values in the center of $M(A)$. A  $G$-covariant
representation $\pi$ of $A$ \cite{dopl,naud} is a representation
$\pi$ of $A$ (and, consequently, $M(A)$) in a Hilbert space $E$
together with a projective representation of $G$ by unitary
operators $U(g)$, $g\in G$, in $E$ such that
\be \pi(g(a)) = U(g)\pi(a)U^*(g),
\qquad U(g)U(g')=\pi(\xi(g,g'))U(gg').
\ee

\section{Spontaneous symmetry breaking}

Given a topological $^*$-algebra $A$ and its state $f$, let $\rho$
be an automorphism of $A$. Then it defines a state
\mar{0810}\beq
f_\rho(a)=f(\rho(a)), \qquad a\in A, \label{0810}
\eeq
of $A$. A state $f$ is said to be stationary with respect to an
automorphism $\rho$ of $A$ if
\mar{082}\beq
f(\rho(a))=f(a), \qquad a\in A. \label{082}
\eeq

One speaks about spontaneous symmetry breaking if a state $f$ of a
quantum algebra $A$ fails to be stationary with respect to some
automorphisms of $A$.

We can say something if $A$ is a $C^*$-algebra and its GNS
representations are considered \cite{brat,dixm,book05}.

\begin{theorem} \mar{t1} \label{t1} Let $f$ be a state of a $C^*$-algebra $A$ and
$(\pi_f,\thh_f,E_f)$ the corresponding cyclic representation of
$A$. An automorphism $\rho$ of $A$ defines the state $f_\rho$
(\ref{0810}) of $A$ such that a carrier space $E_{\rho f}$ of the
corresponding cyclic representation $\pi_{\rho f}$ is  isomorphic
to $E_f$.
\end{theorem}

It follows that the representations $\pi_{\rho f}$ can be given by
operators $\pi_{\rho f}(a)=\pi_f(\rho(a))$ in the carrier space
$E_f$ of the representation $\pi_f$, but these representations
fail to be equivalent, unless an automorphism $\rho$ possesses the
unitary representation (\ref{081}) in $E_f$. In this case, a state
$f$ is stationary relative to $\rho$. The converse also is true.

\begin{theorem} \mar{t2} \label{t2}
If a state $f$ of a $C^*$-algebra $A$ is the stationary state
(\ref{082}) with respect to an automorphism $\rho$ of $A$, there
exists a unique unitary representation $U_\rho$ (\ref{0810}) of
$\rho$ in $E_f$ such that
\mar{087}\beq
U_\rho\thh_f=\thh_f. \label{087}
\eeq
\end{theorem}

It follows from Theorem \ref{spr591} that, since any uniformly
continuous one-parameter group of automorphisms of a $C^*$-algebra
$A$ admits a unitary representation, each state $f$ of $A$ is
stationary for this group. However, this is not true for an
arbitrary uniformly continuous group $G$ of automorphisms of $A$.
For instance, let $B(E)$ be the $C^*$-algebra of all bounded
operators in a Hilbert space $E$. Any automorphisms of $B(E)$ is
inner and, consequently, possesses a unitary representation in
$E$. Since the commutant of $B(E)$ reduces to scalars, the group
of automorphisms of $B(E)$ admits a projective representation in
$E$, but it need not be unitary.

It follows from Theorem \ref{spr759}) that, if a state $f$ of a
$C^*$-algebra $A$ is stationary under a strongly continuous group
$G(\mathbb R)$ of automorphisms of $A$, i.e., $f(]dl(a))=0$, there
exists unitary representation of this group in $E_f$. However,
this condition is sufficient, but not necessary.

Moreover, one can show the following \cite{dixm}.

\begin{theorem} \mar{t4} \label{t4}
Let $G$ be a strongly continuous group of automorphisms of a
$C^*$-algebra $A$, and let a state $f$ of $A$ be stationary for
$G$. Then there exists a unique unitary representation of $G$ in
$E_f$ whose operators obey the equality (\ref{087}).
\end{theorem}

Let now $G$ be a group of strongly or uniformly continuous group
of automorphisms of a $C^*$-algebra $A$, and let $f$ be a state of
$A$. Let us consider a set of states $f_g$ (\ref{0810}), $g\in G$,
of $A$ defined by automorphisms $g\in G$. Let $f$ be stationary
with respect to a proper subgroup $H$ of $G$. Then a set of
equivalence classes of states $f_g$, $g\in G$, is a subset of the
factor space $G/H$, but need not coincide with $G/H$.

This is just the case of spontaneous symmetry breaking in Standard
Model where A Higgs vacuum is a stationary state with respect to
some proper subgroup of a symmetry group \cite{nov,SM}.

In axiomatic QFT, the spontaneous symmetry breaking phenomenon is
described by the Goldstone theorem \cite{bog}.

Let $G$ be a connected Lie group of internal symmetries
(automorphisms of the Borchers algebra $A$ over $\id \mathbb R^4$)
whose infinitesimal generators are given by conserved currents
$j^k_\m$. One can show the following \cite{bog}.

\begin{theorem} \mar{t5} \label{t5}
A group $G$ of internal symmetries possesses a unitary
representation in $E_W$ iff the Wightman functions are
$G$-invariant.
\end{theorem}

\begin{theorem} \mar{t6} \label{t6}
A group $G$ of internal symmetries admits a unitary representation
if a strong spectrum condition holds, i.e., there exists a mass
gap.
\end{theorem}

As a consequence, we come to the above mentioned Goldstone
theorem.

\begin{theorem} \mar{t7} \label{t7}
If there is a group $G$ of internal symmetries which are
spontaneously broken, there exist elements $\f\in E_W$ of zero
spin and mass such that $\lng \f| j^k_\m\psi_0\rng\neq 0$ for some
generators of $G$.
\end{theorem}

These elements of unit norm are called Goldstone states. It is
easily observed that, if a group $G$ of spontaneously broken
symmetries contains a subgroup of exact symmetries $H$, the
Goldstone states carrier out a homogeneous representation of $G$
isomorphic to the quotient $G/H$.

This fact attracted great attention to such kind representations
and motivated to describe classical Higgs fields as sections of a
fibre bundle with a typical fibre $G/H$ \cite{higgs,sard08a,tmp}.

\section{Appendixes}

This Section summarizes some relevant material on topological
vector spaces and measures on non-compact spaces.

\subsection{Topological vector spaces}

There are several standard topologies introduced on an
(infinite-dimensional) complex or real vector space and its dual
\cite{book05,rob}. Topological vector spaces throughout are
assumed to be locally convex. Unless otherwise stated, by the dual
$V'$ of a topological vector space $V$ is meant its topological
dual, i.e., the space of continuous linear maps of $V\to \mathbb
R$.

Let us note that a topology on a vector space $V$ often is
determined by a set of seminorms. A non-negative real function $p$
on $V$ is called the seminorm if it satisfies the conditions
\be
p(\la x)=|\la|p(x), \qquad p(x+y)\leq p(x) +p(y), \qquad x,y\in V,
\quad \la\in\mathbb R.
\ee
A seminorm $p$ for which $p(x)=0$ implies $x=0$ is called the
norm. Given any set $\{p_i\}_{i\in I}$ of seminorms on a vector
space $V$, there is the coarsest topology on $V$ compatible with
the algebraic structure such that all seminorms $p_i$ are
continuous. It is a locally convex topology whose base of closed
neighborhoods consists of sets
\be
\{x\, :\, \op\sup_{1\leq i\leq n} p_i(x)\leq \ve\}, \qquad \ve>0,
\qquad n\in\mathbb N_+.
\ee

Let $V$ and $W$ be two vector spaces whose Cartesian product
$V\times W$ is provided with a bilinear form $\lng v,w\rng$ which
obeys the following conditions:

$\bullet$ for any element $v\neq 0$ of $V$, there exists an
element $w\in W$ such that $\lng v, w\rng\neq 0$;

$\bullet$ for any element $w\neq 0$ of $W$, there exists an
element $v\in V$ such that $\lng v, w\rng\neq 0$.

\noindent Then one says that $(V,W)$ is a dual pair. If $(V,W)$ is
a dual pair, so is $(W,V)$. Clearly, $W$ is isomorphic to a vector
subbundle of the algebraic dual $V^*$ of $V$, and $V$ is a
subbundle of the algebraic dual of $W$.

Given a dual pair $(V,W)$, every vector $w\in W$ defines a
seminorm $p_w=|\lng v,w\rng|$ on $V$. The coarsest topology
$\si(V,W)$ on $V$ making all these seminorms continuous is called
the weak topology determined by $W$ on $V$.  It also is the
coarsest topology on $V$ such that all linear forms in $W\subset
V^*$ are continuous. Moreover, $W$ coincides with the dual $V'$ of
$V$ provided with the weak topology $\si(V,W)$, and $\si(V,W)$ is
the coarsest topology on $V$ such that $V'=W$.  Of course, the
weak topology is Hausdorff.

For instance, if $V$ is a Hausdorff topological vector space with
the dual $V'$, then $(V,V')$ is a dual pair. The weak topology
$\si(V,V')$ on $V$ is coarser than the original topology on $V$.
Since $(V',V)$ also is a dual pair, the dual $V'$ of $V$ can be
provided with the weak$^*$ topology topology $\si(V',V)$. Then $V$
is the dual of $V'$, equipped with the weak$^*$ topology.

The weak$^*$ topology is the coarsest case of a topology of
uniform convergence on $V'$. A subset $M$ of a vector space $V$ is
said to absorb a subset $N\subset V$ if there is a number $\e\geq
0$ such that $N\subset \la M$ for all $\la$ with $|\la|\geq \e$.
An absorbent set is one which absorbs all points. A subset $N$ of
a topological vector space $V$ is called bounded if it is absorbed
by any neighborhood of the origin of $V$. Let $(V,V')$ be a dual
pair and $\cN$ some family of weakly bounded subsets of $V$. Every
$N\subset\cN$ yields a seminorm
\be
p_N(v')=\op\sup_{v\in N} |\lng v, v'\rng|
\ee
on the dual $V'$ of $V$. A topology on $V'$ defined by a set of
seminorms $p_N$, $N\in\cN$, is called the topology of uniform
convergence on the sets of $\cN$. When $\cN$ is a set of all
finite subsets of $V$, we have the coarsest topology of uniform
convergence which is the above mentioned weak$^*$ topology
$\si(V',V)$. The finest topology of uniform convergence is
obtained by taking $\cN$ to be a set of all weakly bounded subsets
of $V$. It is called the strong topology. The dual $V''$ of $V'$,
provided with the strong topology, is called the bidual. One says
that $V$ is reflexive if $V=V''$.

Since $(V',V)$ is a dual pair, a vector space $V$ also can be
provided with the topology of uniform convergence on subsets of
$V'$, e.g., the weak$^*$ and strong topologies. Moreover, any
Hausdorff locally convex topology on $V$ is a topology of uniform
convergence. The coarsest and finest topologies of them are the
weak$^*$ and strong topologies, respectively. There is the
following chain
\be
\mathrm{weak}^* \,<\, \mathrm{weak}\,<\, \mathrm{original}\,<\,
\mathrm{strong}
\ee
of topologies on $V$, where $<$ means "to be finer".

For instance, let $V$ be a normed space. The dual $V'$ of $V$ also
is equipped with a norm
\mar{spr446}\beq
\|v'\|'=\op\sup_{\|v\|=1}|\lng v,v'\rng|, \qquad v\in V, \qquad
v'\in V'. \label{spr446}
\eeq
Let us consider a set of all balls $\{v\,:\, \|v\|\leq\e,
\,\e>0\}$ in $V$. The topology of uniform convergence on this set
coincides with strong and normed topologies on $V'$ because weakly
bounded subsets of $V$ also are bounded by a norm. Normed and
strong topologies on $V$ are equivalent. Let $\ol V$ denote the
completion of a normed space $V$. Then $V'$ is canonically
identified to $(\ol V)'$ as a normed space, though weak$^*$
topologies on $V'$ and $(\ol V)'$ are different. Let us note that
both $V'$ and $V''$ are Banach spaces. If $V$ is a Banach space,
it is closed in $V''$ with respect to the strong topology on $V''$
and dense in $V''$ equipped with the weak$^*$ topology. One
usually considers the weak$^*$, weak and normed (equivalently,
strong) topologies on a Banach space.

It should be noted that topology on a finite-dimensional vector
space is locally convex and Hausdorff iff it is determined by the
Euclidean norm.

Let us say a few words on morphisms of topological vector spaces.

A linear morphism between two topological vector spaces is called
the weakly continuous morphism if it is continuous with respect to
the weak topologies on these vector spaces. In particular, any
continuous morphism between topological vector spaces is weakly
continuous \cite{rob}.

A linear morphism between two topological vector spaces is called
 bounded if the image of a bounded set is bounded. Any continuous
morphism is bounded. A topological vector space is called  the
Mackey space if any bounded endomorphism of this space is
continuous (we follow the terminology of \cite{rob}). Metrizable
and, consequently, normed spaces are of this type.

Any linear morphism $\g: V\to W$ of topological vector spaces
yields the dual morphism $\g': W'\to V'$ of the their topological
duals such that
\be
\lng v,\g'(w)\rng= \lng \g(v),w\rng, \qquad v\in V, \qquad w\in W.
\ee
If $\g$ is weakly continuous, then $\g'$ is weakly$^*$ continuous.
If $V$ and $W$ are normed spaces, then any weakly continuous
morphism $\g: V\to W$ is continuous and strongly continuous. Given
normed topologies on $V'$ and $W'$, the dual morphism $\g':W'\to
V'$ is continuous iff $\g$ is continuous.

\subsection{Hilbert, countably Hilbert and nuclear spaces}

Let us recall the relevant basics on pre-Hilbert and Hilbert
spaces \cite{bourb,book05}.

A Hermitian form on a complex vector space $E$ is defined as a
sesquilinear form $\lng.|.\rng$ such that
\be
\lng e|e'\rng=\ol{\lng e'|e\rng}, \qquad \lng \la e|e'\rng=\lng
e|\ol\la e'\rng=\la \lng e|e'\rng, \quad e,e'\in E, \quad \la\in
\mathbb C.
\ee

\begin{remark}
There exists another convention where $\lng e|\la e'\rng=\la \lng
e|e'\rng$.
\end{remark}

A Hermitian form $\lng.|.\rng$ is said to be positive  if $\lng
e|e\rng\geq 0$ for all $e\in E$. All Hermitian forms throughout
are assumed to be positive. A Hermitian form is called
non-degenerate if the equality $\lng e|e\rng= 0$ implies $e=0$. A
complex vector space endowed with a Hermitian form is called the
pre-Hilbert space. Morphisms of pre-Hilbert spaces, by definition,
are isometric.

A Hermitian form provides $E$ with the topology defined by a
seminorm $\| e\|=\lng e|e\rng^{1/2}$. Hence, a pre-Hilbert space
is Hausdorff iff a Hermitian form $\lng.|.\rng$ is non-degenerate,
i.e., a seminorm $\| e\|$ is a norm. In this case, it is called
the scalar product.

A complete Hausdorff pre-Hilbert space is called the Hilbert
space. Any Hausdorff pre-Hilbert space can be completed to a
Hilbert space. .

The following are the standard constructions of new Hilbert spaces
from the old ones.

$\bullet$ Let  $(E^\iota, \lng.|.\rng_{E^\iota})$ be a set of
Hilbert spaces and $\sum E^\iota$ denote a direct sum of vector
spaces $E^\iota$. For any two elements $e=(e^\iota)$ and
$e'=(e'^\iota)$ of $\sum E^\iota$, a sum
\mar{spr403}\beq
\lng e|e'\rng_\oplus = \op\sum_\iota \lng
e^\iota|e'^\iota\rng_{E^\iota} \label{spr403}
\eeq
is finite, and defines a non-degenerate Hermitian form on $\sum
E^\iota$. The completion $\oplus E^\iota$ of $\sum E^\iota$ with
respect to this form is a Hilbert space, called the Hilbert sum of
$E^\iota$.

$\bullet$ Let $(E,\lng.|.\rng_E)$ and $(H,\lng.|.\rng_H)$ be
Hilbert spaces. Their tensor product $E\ot H$ is defined as the
completion of a tensor product of vector spaces $E$ and $H$ with
respect to the scalar product
\be
&& \lng w_1|w_2\rng_\ot=\op\sum_{\iota,\bt}\lng e_1^\iota | e_2^\bt\rng_E
\lng h_1^\iota | h_2^\bt\rng_H,\\
&& w_1=\op\sum_\iota e^\iota_1\ot h^\iota_1,  \quad
w_2=\op\sum_\bt e^\bt_2\ot h^\bt_2, \quad e_1^\iota,e_2^\bt\in E,
\quad h_1^\iota,h_2^\bt\in H.
\ee

$\bullet$ Let $E'$ be the topological dual of a Hilbert space $E$.
Then the assignment
\mar{spr400}\beq
e\to \ol e(e')= \lng e'|e\rng, \qquad e,e'\in E, \label{spr400}
\eeq
defines an antilinear  bijection of $E$ onto $E'$, i.e., $\ol{\la
e}=\ol\la\ol e$. The dual $E'$ of a Hilbert space is a Hilbert
space provided with the scalar product $\lng \ol e|\ol e'\rng'=
\lng e'|e\rng$ such that the morphism (\ref{spr400}) is isometric.
The $E'$ is called the dual Hilbert space,  and is denoted by $\ol
E$.

Physical applications of Hilbert spaces are limited by the fact
that the dual of a Hilbert space $E$ is anti-isomorphic to $E$.
The construction of a rigged Hilbert space describes the dual
pairs $(E,E')$ where $E'$ is larger than $E$ \cite{gelf64}.

Let a complex vector space $E$ have a countable set of
non-degenerate Hermitian forms $\lng.|.\rng_k$, $k\in\mathbb N_+,$
such that
\be
\lng e|e\rng_1\leq \cdots\leq \lng e|e\rng_k\leq\cdots
\ee
for all $e\in E$. The family of norms
\mar{spr445}\beq
\|.\|_k=\lng.|.\rng^{1/2}_k, \qquad k\in\mathbb N_+,
\label{spr445}
\eeq
yields a Hausdorff topology on $E$. A space $E$ is called the
countably Hilbert space if it is complete with respect to this
topology \cite{gelf64}. For instance, every Hilbert space is a
countably Hilbert space where all Hermitian forms $\lng.|.\rng_k$
coincide.

Let $E_k$ denote the completion of $E$ with respect to the norm
$\|.\|_k$ (\ref{spr445}). There is the chain of injections
\mar{1086}\beq
E_1\supset E_2\supset \cdots E_k\supset \cdots \label{1086}
\eeq
together with a homeomorphism $E=\op\cap_k E_k$. The dual spaces
form the increasing chain
\mar{1087}\beq
E'_1\subset E'_2\subset \cdots \subset E'_k \subset \cdots,
\label{1087}
\eeq
and $E'=\op\cup_k E'_k$. The dual $E'$ of $E$ can be provided with
the weak$^*$ and strong topologies. One can show that a countably
Hilbert space is reflexive.

Given a countably Hilbert space $E$ and $m\leq n$, let $T^n_m$ be
a prolongation of the map
\be
E_n\supset E\ni e\to e\in E\subset E_m
\ee
to a continuous map of $E_n$ onto a dense subset of $E_m$. A
countably Hilbert space $E$ is called the nuclear space if, for
any $m$, there exists $n$ such that $T^m_n$ is a nuclear map,
i.e.,
\be
T^n_m(e)=\op\sum_i\la_i\lng e|e^i_n\rng_{E_n} e^i_m,
\ee
where: (i) $\{e^i_n\}$ and $\{e_m^i\}$ are bases for the Hilbert
spaces $E_n$ and $E_m$, respectively, (ii) $\la_i\geq 0$, (iii)
the series $\sum \la_i$ converges \cite{gelf64}.

An important property of nuclear spaces is that they are perfect,
i.e., every bounded closed set in a nuclear space is compact. It
follows immediately that a Banach (and Hilbert) space is not
nuclear, unless it is finite-dimensional. Since a nuclear space is
perfect, it is separable, and the weak$^*$ and strong topologies
(and, consequently, all topologies of uniform convergence) on a
nuclear space $E$ and its dual $E'$ coincide.

Let $E$ be a nuclear space, provided with still another
non-degenerate Hermitian form $\lng.|.\rng$ which is separately
continuous, i.e., continuous with respect to each argument. It
follows that there exist numbers $M$ and $m$ such that
\mar{1085}\beq
\lng e|e\rng\leq M\|e\|_m, \qquad e\in E. \label{1085}
\eeq
Let $\wt E$ denote the completion of $E$ with respect to this
form. There are the injections
\mar{spr450}\beq
E\subset \wt E\subset E', \label{spr450}
\eeq
where $E$ is a dense subset of $\wt E$ and $\wt E$ is a dense
subset of $E'$, equipped with the weak$^*$ topology. The triple
(\ref{spr450}) is called the rigged Hilbert space. Furthermore,
bearing in mind the chain of Hilbert spaces (\ref{1086}) and that
of their duals (\ref{1087}), one can convert the triple
(\ref{spr450}) into the chain of spaces
\mar{1088}\beq
E\subset\cdots\subset E_k\subset\cdots E_1\subset\wt E\subset
E'_1\subset \cdots\subset E'_k\subset\cdots\subset E'.
\label{1088}
\eeq

\begin{remark}
Real Hilbert, countably Hilbert, nuclear and rigged Hilbert spaces
are similarly described.
\end{remark}

\subsection{Measures on locally compact spaces}

Measures on a locally compact space $X$ are defined as continuous
forms on spaces of continuous real (or complex) functions of
compact support on $X$, and they are extended to a wider class of
functions on $X$ \cite{bourb6,book05}.

Let $\mathbb C^0(X)$ be a ring of continuous complex functions on
$X$. For each compact subset $K$ of $X$,  we have a seminorm
\mar{spr490}\beq
p_K(f)=\op\sup_{x\in K} |f(x)| \label{spr490}
\eeq
on $C^0(X)$. These seminorms provide $\mathbb C^0(X)$ with the
Hausdorff topology of compact convergence. We abbreviate with
$\cK(X,\mathbb C)$ the dense subspace of $\mathbb C^0(X)$ which
consists of continuous complex functions of compact support on
$X$. It is a Banach space with respect to a norm
\mar{spr508}\beq
\|f\|=\op\sup_{x\in X} |f(x)|. \label{spr508}
\eeq
Its normed topology, called the topology of uniform convergence,
is finer than the topology of compact convergence, and these
topologies coincide if $X$ is a compact space.

A space $\cK(X,\mathbb C)$ also can be equipped with another
topology, which is especially relevant to integration theory. For
each compact subset $K\subset X$, let $\cK_K(X,\mathbb C)$ be a
vector subspace of $\cK(X,\mathbb C)$ consisting of functions of
support in $K$. Let $\cU$ be a set of all absolutely convex
absorbent subsets $U$ of $\cK(X,\mathbb C)$ such that, for every
compact $K$, a set $U\cap \cK_K(X,\mathbb C)$ is a neighborhood of
the origin in $\cK_K(X,\mathbb C)$ under the topology of uniform
convergence on $K$. Then $\cU$ is a base of neighborhoods for a
(locally convex) topology, called the inductive limit topology, on
$\cK(X,\mathbb C)$ \cite{rob}. This is the finest topology such
that an injection $\cK_K(X,\mathbb C)\to \cK(X,\mathbb C)$ is
continuous. The inductive limit topology is finer than the
topology of uniform convergence, and these topologies coincide if
$X$ is a compact space. Unless otherwise stated, referring to a
topology on $\cK(X,\mathbb C)$, we will mean the inductive limit
topology.

A complex measure on a locally compact space $X$ is defined as a
continuous form $\m$ on a space $\cK(X,\mathbb C)$ of continuous
complex functions of compact support on $X$. The value $\m(f)$,
$f\in \cK(X,\mathbb C)$, is called the integral $\int f\m$ of $f$
with respect to the measure $\m$. The space $M(X,\mathbb C)$ of
complex measures on $X$ is the dual  of $\cK(X, \mathbb C)$. It is
provided with the weak$^*$ topology.

Given a complex measure $\m$, any continuous complex function
$h\in \mathbb C^0(X)$ on $X$ defines the continuous endomorphism
$f\to hf$ of the space $\cK(X,\mathbb C)$ and yields a new complex
measure $h\m(f)=\m(hf)$. Hence, the space $M(X,\mathbb C)$ of
complex measures on $X$ is a module over the ring $\mathbb
C^0(X)$.

Let $\cK(X)\subset \cK(X,\mathbb C)$ denote a vector space of
continuous real functions of compact support on $X$. The
restriction of a complex measure $\m$ on $X$ to $\cK(X)$ is a
continuous complex form on $\cK(X)$, equipped with the inductive
limit topology. Any complex measure $\m$ is uniquely determined as
\mar{spr512}\beq
\m(f)=\m(\re f) + i\m(\im f), \qquad f\in\cK(X,\mathbb C),
\label{spr512}
\eeq
by its restriction to $\cK(X)$. A complex measure $\m$ on $X$ is
called a real measure if its restriction to $\cK(X)$ is a real
form. A complex measure $\m$ is real iff $\m=\ol\m$, where $\ol\m$
is the conjugate measure given by the condition
$\ol\m(f)=\ol{\m(\ol f)}$, $f\in \cK(X,\mathbb C)$.

A  measure on a locally compact space $X$ is defined as a
continuous real form on a space $\cK(X)$ of continuous real
functions of compact support on $X$ provided with the inductive
limit topology. Any real measure on $\cK(X,\mathbb C)$ restricted
to $\cK(X)$ is a measure. Conversely, each measure $\m$ on
$\cK(X)$ is extended to the real measure (\ref{spr512}) on
$\cK(X,\mathbb C)$. Thus, measures on $\cK(X)$ and  real measures
on $\cK(X,\mathbb C)$ can be identified.

A measure $\m$ on a locally compact space $X$ is called  positive
if $\m(f)\geq 0$ for all positive functions $f\in \cK(X)$. Any
measure $\m$ defines the positive measure $|\m|(f)=|\m(f)|$, and
can be represented by the combination
\be
\m=\frac12(|\m|+\m)-\frac12(|\m|-\m)
\ee
of two positive measures.

A complex measure $\m$ on a locally compact space $X$ is called
bounded if there is a positive number $\la$ such that $|\m(f)|\leq
\la\|f\|$ for all $f\in \cK(X, \mathbb C)$. A complex measure $\m$
is bounded iff it is continuous with respect to the topology of
uniform convergence on $\cK(X,\mathbb C)$. Hence, a space
$M^1(X,\mathbb C)\subset M(X,\mathbb C)$ of bounded complex
measures is the dual of $\cK(X,\mathbb C)$, provided with this
topology. It is a Banach space with respect to a norm
\mar{spr565}\beq
\|\m\|=\sup\{ |\m(f)|\,:\, \|f\|=1\},\, f\in \cK(X,\mathbb C)\}.
\label{spr565}
\eeq
Of course, any complex measure on a compact space is bounded. If
$\m$ is a bounded complex measure and $h$ is a bounded continuous
function on $X$, the complex measure $h\m$ is bounded.

Similarly, a Banach space $M^1(X)$ of bounded measures on $X$ is
defined.

\begin{example} \label{spr492} \mar{spr492}
Given a point $x\in X$, the assignment $\ve_x:f\to f(x)$, $f\in
\cK(X)$, defines the Dirac measure on $X$. Any finite linear
combination of Dirac measures is a measure, called a point
measure. The Dirac measure $\ve_x$ is bounded, and $\|\ve_x\|=1$.
\end{example}

Now we extend a class of integrable functions as follows. Let
$\mathbb R_{\pm\infty}$ denote the extended real line, obtained
from $\mathbb R$ by the adjunction of points $\{+\infty\}$ and
$\{-\infty\}$. It is a ring such that $0\cdot\infty=0$ and
$\infty-\infty=0$. Let $J_+$ be a space of positive lower
semicontinuous functions on $X$ which take their values in the
extended real line $\mathbb R_{\pm\infty}$. These functions
possess the following important properties:

$\bullet$ the upper bound of any set of elements of $J_+$ and the
lower bound of a finite set of elements of $J_+$ also are elements
of $J_+$;

$\bullet$ any function $f\in J_+$ is an upper bound of a family of
positive functions $h\in\cK(X)$ such that $h\leq f$.

The last fact enables one to define the upper integral of a
function $f\in J_+$ with respect to a positive measure $\m$ on $X$
as the element
\mar{spr493}\beq
\m^*(f)=\int^* f\m=\sup\{\m(h)\,:\, h\in\cK(X),\,0\leq h\leq f\}
\label{spr493}
\eeq
of $\mathbb R_{\pm\infty}$. Of course, $\m^*(f)=\m(f)$ if $f\in
\cK(X)$.

\begin{example}
Let $U$ be an open subset of $X$ and $\vf_U$ its  characteristic
function. It is readily observed that $\vf_U\in J_+$. Given a
positive measure $\m$ on $X$, the upper integral
$\m(U)=\m^*(\vf_U)$ is called the outer measure of $U$. For
instance, the outer measure of a relatively compact set $U$ (i.e.,
$U$ is a subset of a compact set) is finite. The (finite or
infinite) number $\m(X)=\m^*(1)$ is called the  total mass of a
measure $\m$. In particular, a measure on a locally compact space
is bounded iff it has a finite total mass.
\end{example}

Let $f$ be an arbitrary positive $\mathbb R_{\pm\infty}$-valued
function on a locally compact space $X$ (not necessarily lower
semicontinuous). There exist functions $g\in J_+$ such that $g\geq
f$ (e.g., $g=+\infty$). Then the upper integral of $f$ with
respect to a positive measure $\m$ on $X$ is defined as
\mar{spr495}\beq
\m^*(f)=\inf\{\m^*(h)\,:\, h\in J_+,\,  h\geq f\}. \label{spr495}
\eeq

\begin{example}
The outer measure $\m^*(V)=\m^*(\vf_V)$ of an arbitrary subset $V$
of $X$ exemplifies the upper integral (\ref{spr495}). In
particular, one says that $V\subset X$ is a $\m$-null set if
$\m(V)=0$. Two $\mathbb R_{\pm\infty}$-valued functions $f$ and
$f'$ on a locally compact space are called $\m$-equivalent if they
differ from each other only on a $\m$-null set; then
$\m^*(f)=\m^*(f')$. Two positive measures $\m$ and $\m'$ are said
to be equivalent if any compact $\m$-null set also is a $\m'$-null
set, and {\it vice versa}. They coincide if $\m(K)=\m'(K)$ for any
compact set $K\subset X$.
\end{example}

\begin{example}
A real function $f$ on a subset $V\subset X$ is said to be defined
almost everywhere with respect to a positive measure $\m$ on $X$
if the complement $X\setminus V$ of $V$ is a $\m$-null set. For
instance, an $\mathbb R_{\pm\infty}$-valued function $f$ which is
finite almost everywhere on $X$ exemplifies a real function
defined almost everywhere on $X$. Conversely, one can think of a
positive function defined almost everywhere on $X$ as being
$\m$-equivalent to some positive $\mathbb R_{\pm\infty}$-valued
function on $X$.
\end{example}

The following classes of integrable functions (and maps) are
usually considered.

Let $f$ be a map of a locally compact space $X$ to a Banach space
$F$, provided with a norm $|.|$ (e.g., $F$ is $\mathbb R$ or
$\mathbb C$). Given a positive measure $\m$ on $X$, let us define
the positive (finite or infinite) number
\mar{spr496}\beq
N_p(f)=\left[\int^* |f|^p \m\right]^{1/p}, \qquad 1\leq p<\infty.
\label{spr496}
\eeq
Clearly, $N_p(f)=N_p(f')$ if $f$ and $f'$ are $\m$-equivalent maps
on $X$, i.e., if they differ on a $\m$-null subset of $X$. There
is the Minkowski inequality
\mar{spr497}\beq
N_p(f+f')\leq N_p(f) + N_p(f'). \label{spr497}
\eeq

$\bullet$ Let $R^p_F(X,\m)$ be a space of maps $X\to F$ such that
$N_p(f)<+\infty$. In accordance with the Minkowski inequality
(\ref{spr497}), it is a vector space and $N_p$ (\ref{spr496}) is a
seminorm on $R^p_F(X,\m)$. Provided with the corresponding
topology, $R^p_F(X,\m)$ is a complete space, but not necessarily
Hausdorff. A space $\cK(X,F)$ of continuous maps $X\to F$ of
compact support belongs to $R^p_F(X,\m)$.

$\bullet$ A space $\cL^p_F(X,\m)$ is defined as the closure of
$\cK(X,F)\subset R^p_F(X,\m)$. Elements of $\cL^p_F(X,\m)$ are
called  integrable $F$-valued functions of degree $p$.  In
particular, elements of $\cL^1_F(X,\m)$ are called integrable
$F$-valued functions, while those of $\cL^2_F(X,\m)$ are square
integrable $F$-valued functions. Any element of $R^p_F(X,\m)$
which is $\m$-equivalent to an element of $\cL^p_F(X,\m)$ belongs
to $\cL^p_F(X,\m)$. An $F$-valued map defined almost everywhere on
$X$ also is called integrable if it is $\m$-equivalent to an
element of $\cL^p_F(X,\m)$.

$\bullet$ A space $L^p_F(X,\m)$ consists of classes of
$\m$-equivalent integrable $F$-valued maps of degree $p$. One
usually treat elements of this space as $F$-valued functions
without fear of confusion, and call them integrable $F$-valued
functions of degree $p$, too. The $L^p_F(X,\m)$ is a Banach space
with respect to the norm (\ref{spr496}).

There are the following important relations between the spaces
$L^p_F(X,\m)$, $1\leq p<+\infty$.

If $f\in \cL^p_F(X,\m)$, then $|f|^{(p/q)-1}f$ belongs to
$\cL^q_F(X,\m)$ for any $1\leq q<+\infty$, and {\it vice versa}.
Moreover, $f\to |f|^{(1/q)-1}f$ provides a homeomorphism between
topological spaces $\cL^1_F(X,\m)$ and $\cL^q_F(X,\m)$.

Let the numbers $1< p <+\infty$ and $1< q <+\infty$ obey the
condition
\mar{spr503}\beq
p^{-1} + q^{-1} =1. \label{spr503}
\eeq
If $f\in L^p_\mathbb C(X,\m)$ is an integrable complex function on
$X$ of degree $p$ and $f'\in L^q_\mathbb C(X,\m)$ is that of
degree $q$, then $f\ol f'$ is integrable, i.e., belongs to
$L^1_\mathbb C(X,\m)$. In particular, a space $L^2_\mathbb
C(X,\m)$ of square integrable complex functions on a locally
compact space $X$ is a separable Hilbert space with respect to a
scalar product
\mar{spr500}\beq
\lng f|f'\rng= \int f\ol f'\m. \label{spr500}
\eeq

One can say something more in the case of real functions. Let
numbers $p$ and $q$ obey the condition (\ref{spr503}). Any
integrable real function $f\in \cL^q_\mathbb R(X,\m)$ on $X$ of
degree $q$ defines a continuous real form
\mar{spr505}\beq
\f_f:f'\to \int ff'\m \label{spr505}
\eeq
on a space $L^p_\mathbb R(X,\m)$ such that $N_q(f)=\|\f_f\|$.
Conversely, each continuous real form on $L^p_\mathbb R(X,\m)$ is
of type (\ref{spr505}) where $f$ is an element of $\cL^q_\mathbb
R(X,\m)$ whose equivalence class in $L^q_\mathbb R(X,\m)$ is
uniquely defined. As a consequence, there is an isomorphism
between Banach spaces $L^q_\mathbb R(X,\m)$ and $(L^p_\mathbb
R(X,\m))'$, and a Banach space $L^p_\mathbb R(X,\m)$ is reflexive.

\begin{remark} \label{spr520} \mar{spr520}
One can define a space $L^\infty_\mathbb C(X,\m)$ of complex
infinite integrable functions on $X$ as the dual of a Banach space
$L^1_\mathbb C(X,\m)$. In particular, any bounded continuous
function belongs to $L^\infty_\mathbb C(X,\m)$. Let us note that a
space $L^1_\mathbb C(X,\m)$ is not reflexive, i.e., the dual of
$L^\infty_\mathbb C(X,\m)$, provided with the strong topology,
does not coincide with $L^1_\mathbb C(X,\m)$.
\end{remark}

We now turn to the relation between equivalent measures. Let $f$
be a positive $\mathbb R_{\pm\infty}$-valued function on a locally
compact space $X$. Given a positive measure $\m$ on $X$, a
quantity
\mar{spr521}\beq
\wt\m(f)=\op\sup_{K\subset X} \m^*(\vf_Kf), \label{spr521}
\eeq
where $K$ runs through a set of all compact subsets of $X$, is
called the essential upper integral of $f$. Since $\vf_Kf\leq f$
for any compact subset $K$, the inequality $\wt\m(f)\leq \m^*(f)$
holds. In particular, if $V$ is a subset of $X$ and $\wt
\m(\vf_V)=0$, one says that $V$ is a locally $\m$-null set,
 i.e., any point $x\in X$ has a neighborhood $U$ such that $U\cap
V$ is a $\m$-null set. Essential upper integrals coincide with the
upper ones if $X$ is a locally compact space countable at
infinity.

One says that a function on a subset $V$ of $X$ is defined locally
almost everywhere if the complement of $V$ is a locally null set.
A real function $f$ defined locally almost everywhere on $X$ is
called  locally $\m$-integrable if any point $x\in X$ has a
neighborhood $U$ such that $\vf_Uf$ is a $\m$-integrable function
or, equivalently, if $hf$ is a $\m$-integrable function for any
positive function $h\in \cK(X)$ of compact support.

Let $h$ be a locally $\m$-integrable function which is defined and
non-negative almost everywhere on $X$. There is a positive measure
$h\m$ on $X$ which obeys the relation $h\m(f)=\m(hf)$ for any
$f\in \cK(X)$. One says that $h\m$ is a measure with the basis
$\m$ and the density $h$. For instance, if $h$ is a
$\m$-integrable function, then $h\m$ is a bounded measure on $X$.
This construction also is extended to complex functions and
complex measures, seen as compositions of the positive real ones.

\begin{theorem} \label{spr510} \mar{spr510}
Positive measures $\m$ and $\m'$ on a locally compact space $X$
are equivalent iff $\m'=f\m$, where $f$ is a locally
$\m$-integrable function such that $f>0$ locally almost everywhere
on $X$.
\end{theorem}

   The function $f$ in Theorem \ref{spr510} is called the
Radon--Nikodym derivative.  Of course, $\wt\m'(f')=\wt\m(ff')$ for
any positive integrable function $f'$ on $X$.

\subsection{Haar measures}

Let us point out the peculiarities of measures on locally compact
groups \cite{bourb6,book05}.

Let $G$ be a topological group acting continuously on a locally
compact space $X$ on the left, i.e., a map
\mar{spr455}\beq
\g(g): X\ni x\to gx\in X, \qquad g\in G, \label{spr455}
\eeq
is continuous for any $g\in G$, and so is a map $G\ni g\to gx\in
X$ for any $x\in X$. It should be emphasized that a map
\be
G\times X\ni (g,x)\to gx\in X
\ee
need not be continuous.

Let $f$ be a real function on $X$ and $\m$ a measure on $X$. A
group $G$ acts on $f$ and $\m$ by the laws
\be
(\g(g)f)(x)=f(g^{-1}x),\qquad  (\g(g)\m)(f)=\m(\g(g^{-1})f).
\ee
A measure $\g(g)\m$ is the image of a measure $\m$ with respect to
the map (\ref{spr455}).

A measure $\m$ on $X$, subject to the action of a group $G$, is
said to be:

$\bullet$ invariant if $\g(g)\m=\m$ for all $g\in G$;

$\bullet$ relative invariant, if there is a strictly positive
number $\chi(g)$ such that $\g(g)\m=\chi(g)^{-1}\m$ for each $g\in
G$;

$\bullet$ quasi-invariant, if measures $\m$ and $\g(g)\m$ are
equivalent for all $g\in G$.

\noindent A strictly positive function $g\to \chi(g)$ yields a
representation of $G$ in $\mathbb R$. It is called the multiplier
of a measure $\m$.

Let a topological group $G$ act continuously on a locally compact
space $X$ on the right, i.e.,
\be
\tau(g): X\ni x\to xg^{-1}\in X.
\ee
The corresponding transformations of functions and measures on $X$
read
\be
(\tau(g)f)(x)=f(xg),\qquad (\tau(g)\m)(f)=\m(\tau(g^{-1})f).
\ee
Then invariant, relative invariant, and quasi-invariant measures
on $X$ are defined similarly to the case of $G$ acting on $X$ on
the left.

Now let $G$ be a locally compact group acting on itself by left
and right multiplications
\mar{B10}\ben
&& \g(g): q\to gq, \qquad \tau(g):q\to qg^{-1}, \qquad q\in
G,\nonumber\\
&& \g(g_1)\tau(g_2)=\tau(g_2)\g(g_1). \label{B10}
\een
Accordingly, left- and right-invariant measures,  relative left-
and right-invariant measures,  left- and right-quasi-invariant
measures on a group $G$ are defined. Each measure $\m$ on $G$ also
yields the inverse measure $\m^{-1}$ given by a relation
\be
\int f(g)\m^{-1}(g)=\int f(g^{-1})\m(g), \qquad
  f\in \cK(G).
\ee

A positive (non-vanishing) left-invariant measure on a locally
compact group $G$ is called the left Haar measure (or, simply, the
{Haar measure). Similarly, the right Haar measure} is defined.

\begin{theorem}
A locally compact group $G$ admits a unique Haar measure with
accuracy to a number multiplier. The total mass $\m(G)$ of a Haar
measure $\m$ on $G$ is finite iff $G$ is a compact group.
\end{theorem}

Let us choose, once and for all, a left Haar measure $dg$ on a
locally compact group $G$. If $G$ is a compact group, $dg$ is
customarily the Haar measure of total mass 1.

\begin{example}
The Lebesgue measure $dx$ is a Haar measure on the additive group
$G=\mathbb R$. Its inverse is $-dx$.
\end{example}

The equality (\ref{B10}) shows that, if $dg$ is a Haar measure on
$G$, a measure $\tau(g')dg$ for any $g'\in G$ also is
left-invariant. Therefore, there exists a unique continuous
strictly positive function $\Delta(g')$ on $G$ such that
$\tau(g')dg=\Delta(g')dg$, $g'\in G$. It is called the modular
function of $G$. If $dg$ is a left Haar measure, its inverse
$(dg)^{-1}$ is a right Haar measure. There is a relation
$(dg)^{-1} =\Delta(g)^{-1}dg$. If $\Delta(g)=1$, a group $G$ is
called unimodular.  Left and right Haar measures on a unimodular
group differ from each other in a number multiplier. For instance,
compact, commutative, and semisimple groups are unimodular. There
is the following criterion of a unimodular group. If the unit
element of a locally compact group has a compact neighborhood
invariant under inner automorphisms, this group is unimodular.

Measures $\m_1,\ldots,\m_n$ on a locally compact group $G$ are
called mutually contractible if there exists a measure
$\op*_i\m_i$ on $G$ given by a relation
\mar{spr540}\beq
\int f(g)\op*_i\m_i(g)=\int f(g_1\cdots g_n) \m_1(g_1)\cdots
\m_n(g_n), \qquad f\in \cK(G). \label{spr540}
\eeq
It is an image of the product measure $\m_1\cdots\m_n$ on
$\op\times^nG$ with respect to a map
\be
\op\times^nG\ni (g_1,\ldots, g_n) \to g_1\cdots g_n\in G.
\ee
Let $\ve_g$, $g\in G$, be the Dirac measure on $G$. The following
relations hold for all $x,y,z\in G$:

$\bullet$ $\ve_x*\ve_y=\ve_{xy}$;

$\bullet$ $\ve_x*\m=\g(x)\m$ and $\m*\ve_x=\tau(x^{-1})\m$;

$\bullet$ if measures $\la,\m,\nu$ are contractible, the pairs of
measures $\la$ and $\m$, $\m$ and $\nu$, $\la*\m$ and $\nu$, $\la$
and $\m*\nu$ also are contractible, and we have
\be
\la*\m*\nu=(\la*\m)*\nu=\la*(\m*\nu).
\ee

For instance, any two bounded measures on $G$ are contractible,
and a space $M^1(G)$ of these measures is a unital Banach algebra
with respect to the contraction $*$ (\ref{spr540}),
 where the Dirac measure $\ve_\bb$ is the unit element.

One also defines:

$\bullet$ the contraction of a measure $\nu$ and a $dg$-integrable
function $f$ on $G$  as the density of the contraction $\nu*(fdg)$
with respect to a Haar measure $dg$ on $G$;

$\bullet$ the contraction of $dg$-integrable functions $f_1$ and
$f_2$ on $G$ as the density of the contraction $(f_1dg)*(f_2dg)$
with respect to a Haar measure $dg$ on $G$, i.e.,
\mar{spr531}\beq
(f_1*f_2)(g)=\int f_1(q)f_2(q^{-1}g)dq, \qquad q,g\in G.
\label{spr531}
\eeq

\subsection{Measures on infinite-dimensional vector spaces}

Throughout this Section, $E$ denotes a real Hausdorff topological
vector space. Infinite-dimensional topological vector spaces need
not be locally compact, and measures on them are defined as
follows \cite{bourb6,gelf64,book05}. All measures are assumed to
be positive.

Let $N(E)$ denote a set of closed vector subspaces of $E$ of
finite codimension, i.e., a vector subspace $V$ of $E$ belongs to
$N(E)$ iff there exists a finite set $y_1,\ldots,y_n$ of elements
of the dual $E'$ of $E$ such that $V$ consists of $x\in E$ which
obey the equalities $\lng x,y_i\rng=0$, $i=1,\ldots,n$.

A quasi-measure (or a cylinder set measure in the terminology of
\cite{gelf64})
   on $E$ is defined as family
$\m=\{\m_V, V\in N(E)\}$ of bounded measures $\m_V$ on
finite-dimensional vector spaces $E/V$ such that if $W\subset V$,
the measure $\m_V$ is the image of a measure $\m_W$ with respect
to the canonical morphism $E/W\to E/V.$

For instance, each bounded measure on $E$ yields a quasi-measure
$\{\m_V, V\in N(E)\}$, where $\m_V$ is the image of a measure $\m$
with respect to the canonical morphism $r_V:E\to E/V$. There is
one-to-one correspondence between the bounded measures on $E$ and
the quasi-measures on $E$ which obey the following condition. For
any $\ve>0$, there exists a compact subset $K\subset E$ such that
\be
\m_V(E/V-r_V(K))\leq\ve, \qquad  V\in N(E).
\ee
Clearly, any quasi-measure on a finite-dimensional vector space is
a measure.

Let $\g:E\to F$ be a continuous morphism of topological vector
spaces. For any $W\in N(F)$, a subspace $V=\g^{-1}(W)$ of $E$
belongs to $N(E)$, and $\g$ yields a morphism $\g^W: E/V\to F/W$.
Let $\m=\{\m_V,V\in N(E)\}$ be a quasi-measure on $E$. Then one
can assign the measure
\be
\nu_W=\g^W_*(\m_{\g^{-1}(W)})
\ee
to each $W\in N(F)$. It is readily observed that the family
$\nu=\{\nu_W, W\in N(F)\}$  is a quasi-measure on $F$. It is
called the image of a quasi-measure $\m$ with respect to a
continuous morphism $\g$.

In particular, let $F=\mathbb R$, and let $y\in E'$ be a
continuous form on $E$. The image $\m_y$ of a quasi-measure $\m$
on $E$ with respect to a form $y$ is a measure on $\mathbb R$. The
Fourier transform of a quasi-measure $\m$ on $E$ is defined as a
complex function
\mar{G0}\beq
Z(y)=\int_{\mathbb R} e^{it}\m_y(t)  \label{G0}
\eeq
on the dual $E'$ of $E$. If $\m$ is a bounded measure on $E$, its
Fourier transform reads
\mar{G1}\beq
Z(y)=\op\int_E\exp[i\lng x,y\rng]\m(x). \label{G1}
\eeq

Let us point out the following variant of the well-known Bochner
theorem. A complex function $Z$ on a topological vector space $F$
is called positive-definite if
\be
\op\sum_{i,j} Z(y_i-y_j)\ol c_ic_j \geq 0
\ee
for any finite set $y_1,\ldots,y_m$ of elements of $F$ and any
complex numbers $c_1,\ldots,c_m$.

\begin{theorem} \label{box} \mar{box}
The Fourier transform (\ref{G1}) provides a bijection of the set
of quasi-measures on a Hausdorff topological vector space $E$ to
the set of positive-definite functions on the dual $E'$ of $E$
whose restriction to any finite-dimensional subspace of $E'$ is
continuous.
\end{theorem}

For instance, let $M(y)$ be a seminorm on $E'$. Then a function
\mar{spr523}\beq
Z(y)=\exp\left[-\frac12 M(y)\right] \label{spr523}
\eeq
on $E'$ is positive-definite. By virtue of Theorem \ref{box},
there is a unique quasi-measure $\m_M$ on $E$ whose Fourier
transform is $Z(y)$ (\ref{spr523}). It is called the Gaussian
quasi-measure with a covariance form $M$.

\begin{example}
Let $E=\mathbb R^n$ be a finite-dimensional vector space,
coordinated by $(x^i)$, and let $M$ be a norm on the dual of $E$.
A Gaussian measure on $E$ with a covariance form $B$ is equivalent
to the Lebesgue measure on $E$, and reads
\mar{qm610}\beq
\m_M= \frac{\mathrm{det}[M]^{1/2}}{(2\pi)^{n/2}}
\exp\left[-\frac12 (M^{-1})_{ij}x^ix^j\right]d^nx. \label{qm610}
\eeq
\end{example}

\begin{example} \label{gaus} \mar{gaus}
Let $E$ be a Banach space and $E'$ its dual, provided with the
norm $\|.\|'$ (\ref{spr446}). A Gaussian quasi-measure on $E$ with
the covariance form $\|.\|'$ is called canonical. One can show
that this quasi-measure fails to be a measure, unless $E$ is
finite-dimensional. Let $T$ be a continuous operator in $E'$. Then
\mar{spr524}\beq
y\to \|Ty\| \label{spr524}
\eeq
is a seminorm on $E'$. A Gaussian quasi-measure on $E$ with the
covariance form (\ref{spr524}) is proved to be a measure iff $T$
is a Hilbert--Schmidt operator.
\end{example}

Let $E$ be a real nuclear space and $E'$ its dual, equipped with
the topology of uniform convergence. Let us recall that all
topologies of uniform convergence (including weak$^*$ and strong
topologies) on $E'$ coincide, and $E$ is reflexive. A
quasi-measure on $E'$ is a measure iff its Fourier transform on
$E$ (which is the dual of $E'$) is continuous. A variant of the
Bochner theorem for nuclear spaces states the following
\cite{gelf64,book05}.

\begin{theorem} \mar{spr525} \label{spr525}
The Fourier transform
\be
Z(x)=\int\exp[i\lng x,y\rng]\m(y)
\ee
provides a bijection of the set of measures on the dual $E'$ of a
real nuclear space $E$ to the set of continuous positive-definite
functions on $E$
\end{theorem}

\begin{remark} \label{spr526} \mar{spr526}
Let $E\subset\wt E\subset E'$ be a real rigged Hilbert space,
defined by a norm $\|.\|$ on $E$. Let $T$ be a nuclear operator in
$\wt E$ and $\|.\|_T$ the restriction of the seminorm
$y\to\|Ty\|$, $y\in \wt E$, (\ref{spr524}) on $\wt E$ to $E$. Then
the Gaussian measures $\m$ and $\m_T$ on $E'$ with the covariance
forms $\|.\|$ and $\|.\|_T$ are not equivalent. The Gaussian
measures $\m$ and $\m_T$ are equivalent if $T$ is a sum of the
identity and a nuclear operator. In particular, all Gaussian
measures on a finite-dimensional vector space are equivalent.
\end{remark}

\end{document}